\documentclass[a4paper,10pt,reqno]{amsart}

\usepackage{amssymb,amsmath,amscd,amsthm,hyperref}
\usepackage{mathtools}

\title[Bounded Rationality from Elementary Computations]{Bounded rational decision-making from elementary computations that reduce uncertainty}
\author{Sebastian Gottwald, Daniel A. Braun}

\newtheorem{Theorem}{Theorem}[section]

\newtheorem{Proposition}[Theorem]{Proposition}
\theoremstyle{definition}
\newtheorem{Definition}[Theorem]{Definition}
\theoremstyle{remark}

\newtheorem*{Example}{Example}
\numberwithin{equation}{section}

\begin{document}




\begin{abstract}
In its most basic form, decision-making can be viewed as a computational process that progressively eliminates alternatives, thereby reducing uncertainty. Such processes are generally costly, meaning that the amount of uncertainty that can be reduced is limited by the amount of available computational resources. Here, we introduce the notion of elementary computation based on a fundamental principle for probability transfers that reduce uncertainty. Elementary computations can be considered as the inverse of Pigou-Dalton transfers applied to probability distributions, closely related to the concepts of majorization,  T-transforms, and generalized entropies that induce a preorder on the space of probability distributions. As a consequence we can define resource cost functions that are order-preserving and therefore monotonic with respect to the uncertainty reduction. This leads to a comprehensive notion of decision-making processes with limited resources. Along the way, we prove several new results on majorization theory, as well as on entropy and divergence measures.
\end{abstract}

\thanks{This is a pre-print of an article published in Entropy 2019, 21(4), 375. The final authenticated version is available online at: https://www.mdpi.com/1099-4300/21/4/375}

\maketitle

\section{Introduction}

In rational decision theory, uncertainty may have multiple sources that ultimately share the commonality that they reflect a lack of knowledge on the part of the decision-maker about the environment. A paramount example is the perfectly rational decision-maker \cite{Neumann1944} that has a probabilistic model of the environment and chooses its actions so as to maximize the expected utility entailed by the different choices. When we consider bounded rational decision-makers \cite{Simon1955}, we may add another source of uncertainty arising from the decision-maker's limited processing capabilities, since the decision-maker will not only accept a single best choice, but will accept any satisficing option.
Today, bounded rationality is an active research topic that crosses multiple scientific fields such as economics, political science, decision theory, game theory, computer science, and neuroscience \cite{Russel1995,Ochs1995,Lipman1995,Aumann1997,Gigerenzer2001,Mattsson2002,Jones2003,Sims2003,Wolpert2006,Howes2009,Still2009,Tishby2011,Spiegler2011,Kappen2012,Burns2013,Ortega2013,Lewis2014,Acerbi2014,Gershman2015},
where uncertainty is one of the most important common denominators.

Uncertainty is often equated with Shannon entropy in information theory \cite{Shannon1948}, measuring the average number of yes/no-questions that have to be answered to resolve the uncertainty. Even though Shannon entropy has many desirable properties, there are plenty of alternative suggestions for entropy measures in the literature, known as generalized entropies, such as R\'{e}nyi entropy \cite{Renyi1961} or Tsallis entropy \cite{Tsallis1988}. Closely related to entropies are divergence measures, which express how probability distributions differ from a given reference distribution. If the reference distribution is uniform then divergence measures can be expressed in terms of entropy measures, which is why divergences can be viewed as generalizations of entropy, for example the Kullback-Leibler divergence \cite{Kullback1951} generalizing Shannon entropy. 

Here, we introduce the concept of elementary computation based on a slightly stronger notion of uncertainty than is expressed by Shannon entropy, or any other generalized entropy alone, but is equivalent to all of them combined. Equating decision-making with uncertainty reduction, this leads to a new comprehensive view of decision-making with limited resources. Our main contributions can be summarized as follows:

\begin{enumerate}
\item[$(i)$] Based on a fundamental concept of probability transfers related to the Pigou-Dalton principle of welfare economics \cite{Dalton1920}, we promote a generalized notion of uncertainty reduction of a probability distribution that we call \textit{elementary computation}. This leads to a natural definition of \textit{cost functions} that quantify the resource costs for uncertainty reduction necessary for decision-making. We generalize these concepts to arbitrary reference distributions. In particular, we define Pigou-Dalton-type transfers for probability distributions relative to a reference or prior distribution, which induce a preorder that is slightly stronger than Kullback-Leibler divergence, but is equivalent to the notion of divergence given by all $f$-divergences combined. We prove several new characterizations of the underlying concept, known as \textit{relative majorization}. 

\item[$(ii)$] An interesting property of cost functions is their behavior under \textit{coarse-graining}, which plays an important role in decision-making and formalizes the general notion of making abstractions. More precisely, if a decision in a set $\Omega$ is split up into two steps by partitioning $\Omega = \bigcup_i A_i$ and first deciding in the set of (coarse-grained) partitions $\{A_i\}_{i}$ and secondly choosing a fine-grained option inside the selected partition $A_i$, then it is an important question how the cost for the total decision-making process differs from the sum of the costs in each step. We show that $f$-divergences are superadditive with respect to coarse-graining, which means that decision-making costs can potentially be reduced by splitting up the decision into multiple steps. In this regard, we find evidence that the well-known property of Kullback-Leibler divergence of being additive under coarse-graining might be viewed as describing the minimal amount of processing cost that cannot be reduced by a more intelligent decision-making strategy.

\item[$(iii)$] We define \textit{bounded rational} decision-makers as decision-making processes that are optimizing a given utility function under a constraint on the cost function, or minimizing the cost function under a minimal requirement on expected utility. As a special case for Shannon-type information costs, we arrive at information-theoretic bounded rationality, which may form a normative baseline for bounded-optimal decision-making in the absence of process-dependent constraints. We show that bounded-optimal posteriors with informational costs trace a path through probability space that can itself be seen as an anytime decision-making process, where each step optimally trades off utility and processing costs.

\item[$(iv)$]  We show that \textit{Bayesian inference} can be seen as a decision-making process with limited resources given by the number of available datapoints.

\end{enumerate}

Section \ref{sec:1} deals with $(i)$ and $(ii)$, aiming at a general characterization of decision-making in terms of uncertainty reduction. Item $(iii)$ is covered in Section \ref{sec:bounded}, deriving information-theoretic bounded rationality as a special case. Section \ref{sec:example} illustrates the concepts with an example including item $(iv)$. Sections \ref{sec:discussion} and \ref{sec:conclusion} contain a general discussion and concluding remarks.

\bigskip

\section*{Notation}

Let $\mathbb R$ denote the real numbers, $\mathbb R_+ \coloneqq [0,\infty)$ the set of non-negative real numbers, and $\mathbb Q$ the rational numbers. We write $|A|$ for the number of elements contained in a countable set $A$, and $B\setminus A$ for the set difference, that is the set of elements in $B$ that are not in $A$. $\mathbb P_\Omega$ denotes the set of probability distributions on a set $\Omega$, in particular, any $p\in \mathbb P_\Omega$ is normalized so that $p(\Omega) = \mathbb E_p[1] = 1$. Random variables are denoted by capital letters $X,Y,Z$, while their explicit values are denoted by small letters $x,y,z$. For the probability distribution of a random variable $X$ we write $p(X)$, and $p(x)$ for the values of $p(X)$. Correspondingly, the expectation $\mathbb E[f(X)]$ is also written as $\mathbb E_{p(X)}[f(X)]$, $\mathbb E_{p(X)}[f]$, or $\mathbb E_p[f]$. We also write $\langle f\rangle_p := \frac{1}{n}\sum_{i=1}^n f(x_n) $, to denote the approximation of $\mathbb E_{p}[f]$ by an average over samples $\{x_1,\dots,x_n\}$ from $p\in\mathbb P_\Omega$.

\bigskip

\section{Decision-making with limited resources} \label{sec:1}

In this section, we develop the notion of a \textit{decision-making process with limited resources} following the simple assumption that any decision-making process 
\begin{itemize}
\item[$(i)$] reduces uncertainty
\item[$(ii)$] by spending resources.
\end{itemize}
Starting from an intuitive interpretation of \textit{uncertainty} and \textit{resource costs}, these concepts are refined incrementally until a precise definition of a decision-making process is given at the end of this section (Definition \ref{def:DM}) in terms of \textit{elementary computations}. Here, a decision-making process is a comprehensive term that describes all kinds of biological as well as artificial systems that are searching for solutions to given problems, for example a human decision-maker that burns calories while thinking, or a computer that uses electric energy to run an algorithm. However, \textit{resources} do not necessarily refer to a real consumable quantity but can also represent more explicit resources (like time) as a proxy, for example the number of binary comparisons in a search algorithm, the number of forward simulations in a reinforcement learning algorithm, the number of samples in a Monte Carlo algorithm, or, even more abstractly, they can express the limited availability of some source of information, like for example the number of datapoints that are available to an inference algorithm (see Section \ref{sec:example}).

\begin{figure}
\centering
\includegraphics[width=\textwidth]{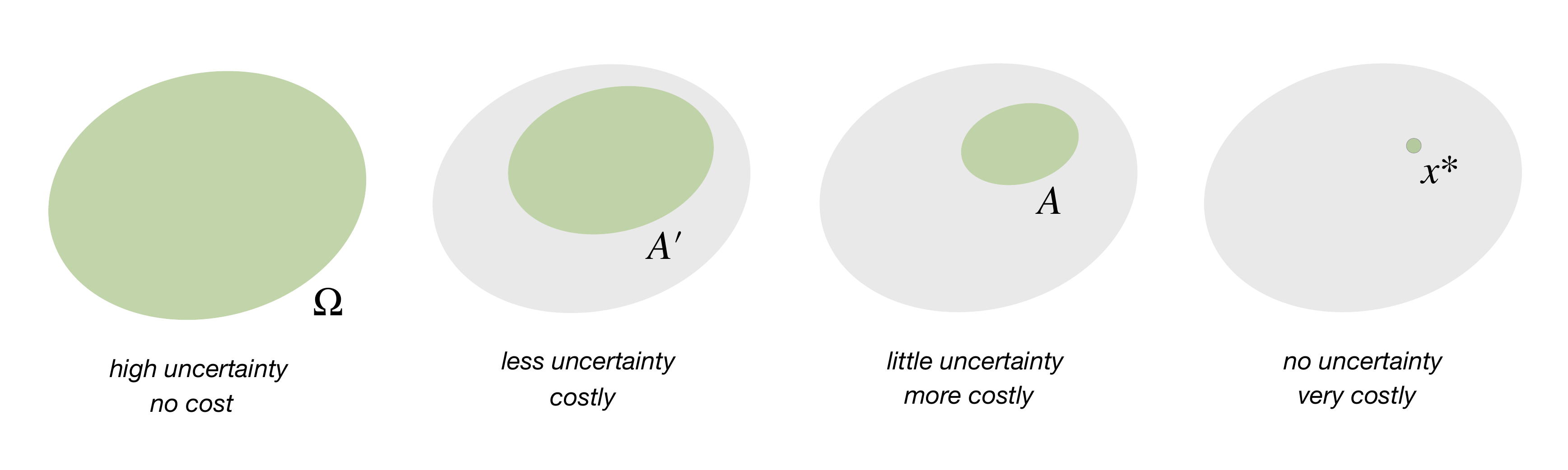}
\caption{Decision-making as search in a set of options. At the expense of more and more resources, the number of uncertain options is progressively reduced until $x^\ast$ is the only remaining option.}
\label{fig:1}
\end{figure}

\medskip
\subsection{Uncertainty reduction by eliminating options} \label{sec:elim}

In its most basic form, the concept of decision-making can be formalized as the process of looking for a decision $x\in \Omega$ in a discrete set of options $\Omega=\{x_1,\dots,x_N\}$. We say that a decision $x\in\Omega$ is \textit{certain}, if repeated queries of the decision-maker will result in the same decision, and it is \textit{uncertain}, if repeated queries can result in different decisions. Uncertainty reduction then corresponds to reducing the amount of uncertain options. Hence, a decision-making process that transitions from a space $\Omega$ of options to a strictly smaller subset $A\subsetneq\Omega$ reduces the amount of uncertain options from $N=|\Omega|$ to $N_A\coloneqq |A|<N$, with the possible goal to eventually find a single certain decision $x^\ast$. Such a process is generally costly, the more uncertainty is reduced the more resources it costs (Figure \ref{fig:1}). The explicit mapping between uncertainty reduction and resource cost depends on the details of the underlying process and on which explicit quantity is taken as the resource. For example, if the resource is given by time (or any monotone function of time), then a search algorithm that eliminates options sequentially until the target value is found (linear search) is less cost efficient than an algorithm that takes a sorted list and in each step removes half of the options by comparing the mid point to the target (logarithmic search). Abstractly, any real-valued function $C$ on the power set of $\Omega$ that satisfies $C(A')<C(A)$ whenever $A\subsetneq A'$ might be used as a cost function in the sense that $C(A)$ quantifies the expenses of reducing the uncertainty from $\Omega$ to $A\subset \Omega$.

In utility theory decision-making is modelled as an optimization process that maximizes a so-called \textit{utility function} $U:\Omega \to \mathbb R$ (which can itself be an \textit{expected} utility with respect to a probabilistic model of the environment, in the sense of von Neumann and Morgenstern \cite{Neumann1944}). A decision-maker that is optimizing a given utility function $U$ obtains a utility of $\frac{1}{N_A} \sum_{x\in A} U(x) \geq \frac{1}{N} \sum_{x\in \Omega} U(x)$ on average after reducing the amount of uncertain options from $N$ to $N_A<N$ (see Figure \ref{fig:2}). A decision-maker that completely reduces uncertainty by finding the optimum $x^\ast = \mathrm{argmax}_{x\in \Omega} U(x)$ is called \textit{rational} (w.l.o.g. we can assume that $x^\ast$ is unique, by redefining $\Omega$ in the case when it is not). Since uncertainty reduction generally comes with a cost, a utility optimizing decision-maker with limited resources,  correspondingly called \textit{bounded rational} (see Section \ref{sec:bounded}), in contrast will obtain only uncertain decisions from a subset $A\subset\Omega$. Such decision-makers seek satisfactory rather than optimal solutions, for example by taking the first option that satisfies a minimal utility requirement, which Herbert Simon calls a \textit{satisficing} solution \cite{Simon1955}.

Summarizing, we conclude that a decision-making process with decision space $\Omega$ that successively eliminates options can be represented by a mapping $\phi$ between subsets of $\Omega$, together with a cost function $C$ that quantifies the total expenses of arriving at a given subset, 
\begin{equation}\label{DM_elim0}
\Omega \longrightarrow \cdots \longrightarrow A' \longrightarrow \phi(A')\longrightarrow \cdots \longrightarrow A
\end{equation}

\vspace{-4pt}
\noindent such that
\begin{equation}\label{DM_elim}
\Omega \supset  A' \supset \phi(A') \supset A,  \qquad 0 = C(\Omega) < C(A') < C(\phi(A')) < C(A) \, ,
\end{equation}

\noindent For example, a rational decision-maker can afford $C(\{x^\ast\})$, whereas a decision-maker with limited resources can typically only afford uncertainty reduction with cost $C(A)<C(\{x^\ast\})$.

\begin{figure}
\centering
\includegraphics[width=\textwidth]{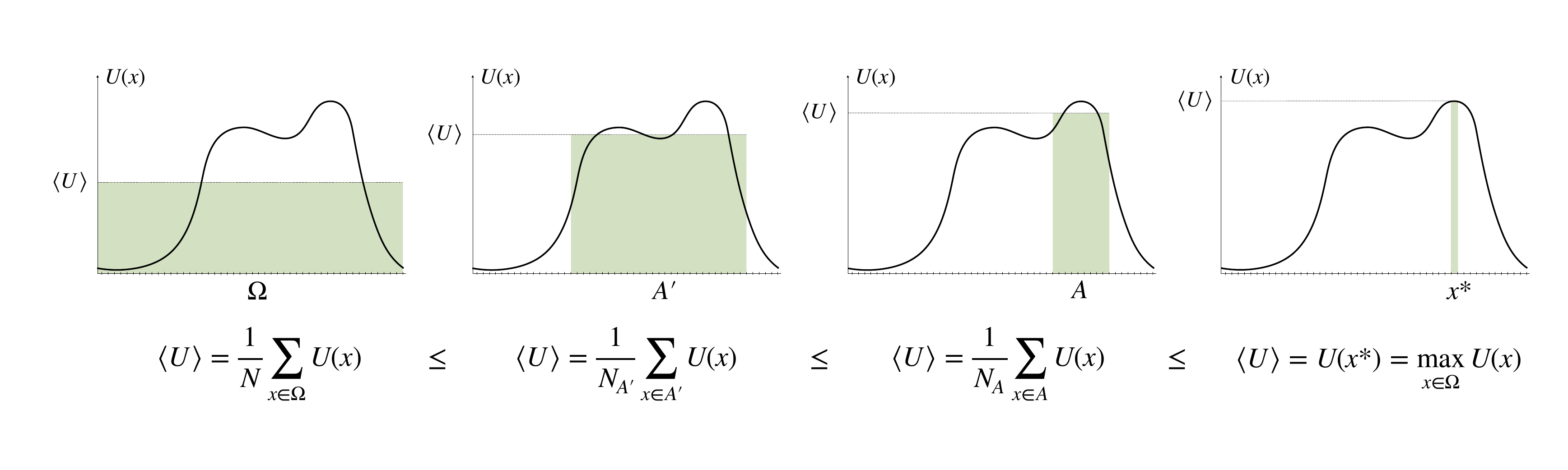}

\vspace{-8pt}
\caption{Decision-making as utility optimization process.}
\label{fig:2}
\end{figure}

From a probabilistic perspective, a decision-making process as described above is a transition from a uniform probability distribution over $N$ options to a uniform probability distribution over $N'<N$ options, that converges to the Dirac measure $\delta_{x^\ast}$ centered at $x^\ast$ in the fully rational limit. From this point of view, the restriction to uniform distributions is artificial. A decision-maker that is uncertain about the optimal decision $x^\ast$ might indeed have a bias towards a subset $A$ without completely excluding other options (the ones in $A^c = \Omega{\setminus} A$), so that the behavior must be properly described by a probability distribution $p\in\mathbb P_\Omega$. Therefore, in the following section, we extend \eqref{DM_elim0} and \eqref{DM_elim} to transitions between probability distributions. In particular, we must replace the power set of $\Omega$ by the space of probability distributions on $\Omega$, denoted by $\mathbb P_\Omega$.

\bigskip
\subsection{Probabilistic decision-making} \label{sec:probDM}

Let $\Omega$ be a discrete decision space of $N\,{=}\,|\Omega|{<}\,\infty$ options, so that $\mathbb P_\Omega$ consists of discrete distributions $p$, often represented by probability vectors $p=(p_1,\dots, p_N)$. However, many of the concepts presented in this and the following section can be generalized to the continuous case \cite{Marshall2011,Joe1990}.

Intuitively, the uncertainty contained in a distribution $p\in\mathbb P_\Omega$ is related to the relative inequality of its entries, the more similar its entries are, the higher the uncertainty. This means that uncertainty is increased by moving some probability weight from a more likely option to a less likely option. It turns out that this simple idea leads to a concept widely known as \textit{majorization} \cite{HLP1934,Arnold1987,Pecaric1992,Bhatia1997,Marshall2011,Arnold2018}, which has roots in the economic literature of the early 19th century \cite{Lorenz1905,Pigou1912,Dalton1920}, where it was introduced to describe income inequality, later known as the \textit{Pigou-Dalton Principle of Transfers}. Here, the operation of moving weight from a more likely to a less likely option corresponds to the transfer of income from one individual of a population to a relatively poorer individual (also known as a \textit{Robin Hood operation} \cite{Arnold1987}). Since a decision-making process can be viewed as a sequence of uncertainty reducing computations, we call the inverse of such a Pigou-Dalton transfer an \textit{elementary computation}.

\begin{Definition}[Elementary computation]\label{def:pigoudalton}
A transformation on $\mathbb P_\Omega$ of the form 
\begin{equation}\label{eq:pigoudalton}
T_\varepsilon: p\mapsto (p_1, \dots,  p_m+\varepsilon, \dots, p_{n} - \varepsilon, \dots, p_N) \, ,
\end{equation}

\noindent where $m,n$ are such that $p_m \leq p_n$, and $0< \varepsilon\leq \frac{p_n{-}p_m}{2}$, is called a Pigou-Dalton transfer (see Figure \ref{fig:1c}). We call its inverse $T^{-1}_\varepsilon$ an \textit{elementary computation}.
\end{Definition}

Since making two probability values more similar or more dissimilar are the only two possibilities to minimally transform a probability distribution, elementary computations are the most basic principle of how uncertainty is reduced. Hence, we conclude that a distribution $p'$ has more uncertainty than a distribution $p$ if and only if $p$ can be obtained from $p'$ by finitely many elementary computations (and permutations, which are not considered an elementary computation due to the choice of $\varepsilon$).

\begin{figure}
\centering
\includegraphics[width=\textwidth]{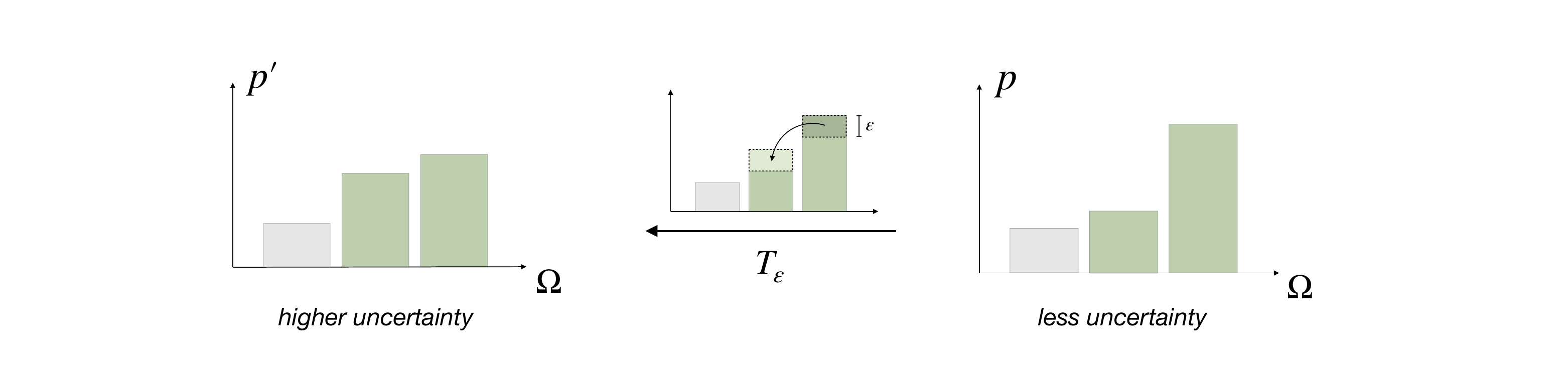}

\vspace{-2pt}
\caption{A Pigou-Dalton transfer as given by Equation \eqref{eq:pigoudalton}. The transfer of probability from a more likely to a less likely option increases uncertainty.}
\label{fig:1c}
\end{figure}

\begin{Definition}[Uncertainty] \label{def:uncertainty}
We say that $p'\in\mathbb P_\Omega$ contains more uncertainty than $p\in\mathbb P_\Omega$, denoted by 
\begin{equation}
p' \prec p \, ,
\end{equation}

\noindent if and only if $p$ can be obtained from $p'$ by a finite number of elementary computations and permutations.
\end{Definition}

\noindent Note that, mathematically this defines a \textit{preorder} on $\mathbb P_\Omega$, i.e. a reflexive ($p\prec p$ for all $p\in\mathbb P_\Omega$) and transitive (if $p''\prec p'$, $p'\prec p$ then $p''\prec p$ for all $p,p',p''\in\mathbb P_\Omega$) binary relation.

In the literature, there are different names for the relation between $p$ and $p'$ expressed by Definition \ref{def:uncertainty}, for example $p'$ is called \textit{more mixed} than $p$ \cite{Ruch1976}, \textit{more disordered} than $p$ \cite{Rossignoli2004}, \textit{more chaotic} than $p$ \cite{Bhatia1997}, or an \textit{average} of $p$ \cite{HLP1934}. Most commonly however, $p$ is said to \textit{majorize} $p'$, which started with the early influences of Muirhead \cite{Muirhead1902}, and Hardy, Littlewood, and P\'{o}lya \cite{HLP1934} and was developed by many authors into the field of majorization theory (a standard reference is by Marshall, Olkin, and Arnold \cite{Marshall2011}), with far reaching applications until today, especially in nonequilibrium thermodynamics and quantum information theory \cite{Brandao2013,Horodecki2013,Gour2015}. 

There are plenty of equivalent (arguably less intuitive) characterizations of $p\prec p'$, some of which we are summarizing below. However, one characterization makes use of a concept very closely related to Pigou-Dalton transfers, known as \textit{T-transforms} \cite{Bhatia1997,Marshall2011}, which expresses the fact that moving some weight from a more likely option to a less likely option is equivalent to taking (weighted) averages of the two probability values. More precisely, a T-transform is a linear operator on $\mathbb P_\Omega$ with a matrix of the form $T = (1-\lambda) \mathbb I + \lambda \Pi$, where $\mathbb I$ denotes the identity matrix on $\mathbb R^N$, $\Pi$ denotes a permutation matrix of two elements, and $0\leq \lambda\leq 1$. If $\Pi$ permutes $p_m$ and $p_n$, then $(Tp)_k = p_k$ for all $k\not \in \{m,n\}$, and 
\begin{equation} \label{Ttransformexplicit}
(Tp)_m = (1-\lambda) p_m + \lambda p_n \, , \quad (Tp)_n = \lambda p_m + (1-\lambda) p_n \, .
\end{equation}

\noindent Hence, a T-transform considers any two probability values $p_m$ and $p_n$ of a given $p\in\mathbb P_\Omega$, calculates their weighted averages with weights $(1-\lambda,\lambda)$ and $(\lambda,1-\lambda)$, and replaces the original values with these averages. From \eqref{Ttransformexplicit} it follows immediately that a T-transform with parameter $0< \lambda \leq \frac{1}{2}$ and a permutation $\Pi$ of $p_m,p_n$ with $p_m\leq p_n$ is a Pigou-Dalton transfer with $\varepsilon = (p_n-p_m)\lambda$. Also allowing $\frac{1}{2}\leq\lambda\leq 1$ means that T-transfers include permutations, in particular, $p'\prec p$ if and only if $p'$ can be derived from $p$ by successive applications of finitely many T-transforms 

Due to a classic result by Hardy, Littlewood and P\'{o}lya \cite[p.49]{HLP1934}, this characterization can be stated in an even simpler form by using \textit{doubly stochastic matrices}, i.e. matrices $A=(A_{ij})_{i,j}$ with $A_{ij}\geq 0$ and $\sum_i A_{ij} = 1  = \sum_j A_{ij}$ for all $i,j$. By writing $xA \coloneqq A^T x$ for all $x\in\mathbb R^N$, and $e \coloneqq (1,\dots,1)$, these conditions are often stated as 
\begin{equation}\label{def:stochmat}
A_{ij} \geq 0 \, , \quad Ae = e \, , \quad eA = e \, .
\end{equation}

\noindent Note that doubly stochastic matrices can be viewed as generalizations of T-transforms in the sense that a T-transform takes an average of two entries, whereas if $p' = pA$ with a doubly stochastic matrix $A$, then $p'_j=\sum_i A_{ij} p_i$ is a convex combination, or a weighted average, of $p$ with coefficients $(A_{ij})_i$ for each $j$. This is also, why $p'$ is then called \textit{more mixed} than $p$ \cite{Ruch1976}. Therefore, similar to T-transforms, we might expect that if $p'$ is the result of an application of a doubly stochastic matrix, $p' = pA$, then $p'$ is an average of $p$ and therefore contains more uncertainty than $p$. This is exactly what is expressed by characterization $(iii)$ in the following theorem. A similar characterization of $p'\prec p$ is that $p'$ must be given by a convex combination of permutations of the elements of $p$ (see property $(iv)$ below). 

Without having the concept of majorization, Schur proved that functions of the form  $ p\mapsto \sum_i f(p_i)$ with a convex function $f$ are monotone with respect to the application of a doubly stochastic matrix \cite{Schur1923} (see property $(v)$ below). Functions of this form are an important class of cost functions for probabilistic decision-makers, as we will discuss in Example \ref{ex:genEnt}.

\begin{figure}
\centering
\includegraphics[width=.8\textwidth]{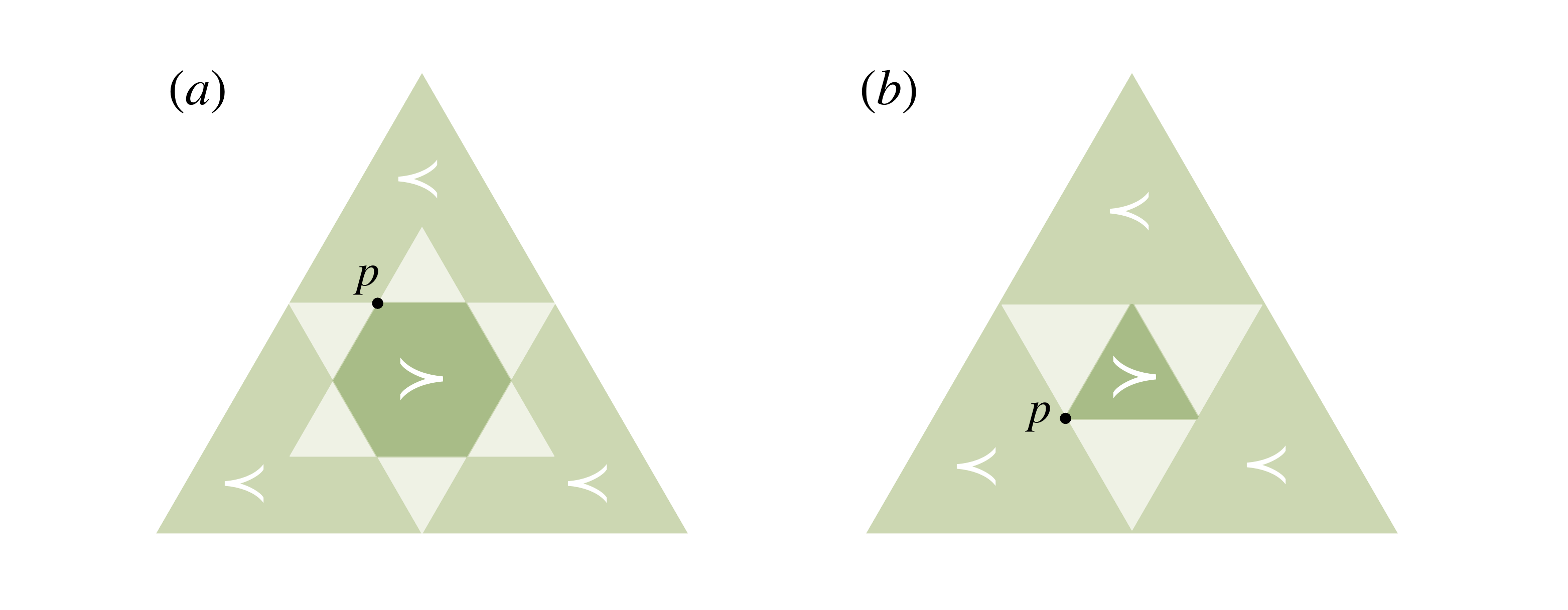}

\vspace{-5pt}
\caption{Comparability of probability distributions in $N=3$. The region in the center consists of all $p'$ that are majorized by $p$, i.e. $p\succ p'$, whereas the outer region consists of all $p'$ that majorize $p$, $p\prec p'$. The bright regions are not comparable to $p$. $(a)$ $p=(\frac{1}{3},\frac{1}{2},\frac{1}{6})$, $(b)$ $p=(\frac{1}{2},\frac{1}{4},\frac{1}{4})$.}
\label{fig:15}
\end{figure}

\begin{Theorem}[Characterizations of $p'\prec p$ \cite{Marshall2011}] \label{thm:characterization}
For $p,p'\in\mathbb P_\Omega$, the following are equivalent:
\begin{enumerate}
\item[$(i)$] $p'\prec p$, i.e. $p'$ contains more uncertainty than $p$ (Definition \ref{def:uncertainty}) \\[-4pt]
\item[$(ii)$] $p'$ is the result of finitely many T-transforms applied to $p$\\[-4pt]
\item[$(iii)$] $p' = pA$ for a doubly stochastic matrix $A$\\[-4pt]
\item[$(iv)$] $p' = \sum_{k=1}^K \theta_k \Pi_k(p)$ where $K\in\mathbb N$, $\sum_{k=1}^K \theta_k = 1$, $\theta_k\geq 0$, and $\Pi_k$ is a permutation for all $k\in\{1,\dots,K\}$\\[-4pt]
\item[$(v)$] $\sum_{i=1}^N f(p'_i) \leq \sum_{i=1}^N f(p_i)$ for all continuous convex functions $f$\\[-4pt]
\item[$(vi)$] $\sum_{i=1}^k (p'_i)^{\downarrow} \leq \sum_{i=1}^k p_i^{\downarrow}$ for all $k\in \{1,\dots,N{-}1\}$, where $p^\downarrow$ denotes the decreasing rearrangement of $p$.
\end{enumerate}
\end{Theorem}

As argued above, the equivalence between $(i)$ and $(ii)$ is straight-forward. The equivalences between $(ii)$, $(iii)$, and $(vi)$ are due to Muirhead \cite{Muirhead1902} and Hardy, Littlewood, and P\'{o}lya \cite{HLP1934}. The implication $(v)\Rightarrow (iii)$ is due to Karamata \cite{Karamata1932} and Hardy, Littlewood, and P\'{o}lya \cite{HLP1929}, whereas $(iii)\Rightarrow (v)$ goes back to Schur \cite{Schur1923}. Mathematically, $(iv)$ means that $p'$ belongs to the convex hull of all permutations of the entries of $p$, and the equivalence $(iii)\Leftrightarrow (iv)$ is known as the Birkhoff-von Neumann theorem. Here, we have stated all relations for probability vectors $p\in\mathbb P_\Omega$, even though they are usually stated for all $p,p'\in\mathbb R^N$ with the additional requirement that $\sum_{i=1}^N p_i = \sum_{i=1}^N p'_i$.

Condition $(vi)$ is the classical and most commonly used definition of majorization \cite{Lorenz1905,HLP1934,Marshall2011}, since it is often the easiest to check in practical examples. For example, from $(vi)$ it immediately follows that uniform distributions over $N$ options contain more uncertainty than uniform distributions over $N'<N$ options, since we have $\sum_{i=1}^k \frac{1}{N} = \frac{k}{N} \leqslant \frac{k}{N'} = \sum_{i=1}^k \frac{1}{N'}$ for all $k<N$, i.e. for $N\geq 3$ it follows that
\begin{equation} \label{precuniform}
\big(\tfrac{1}{N},\dots,\tfrac{1}{N}\big) \prec \big(\tfrac{1}{N-1},\dots,\tfrac{1}{N-1},0\big) \prec \big(\tfrac{1}{2},\tfrac{1}{2},0,\dots,0\big) \prec \big(1,0\dots,0\big) \, .
\end{equation}
In particular, if $A\subset A'\subset \Omega$, then the uniform distribution over $A$ contains less uncertainty than the uniform distribution over $A'$, which shows that the notion of uncertainty introduced in Definition \ref{def:uncertainty} is indeed a generalizatin of the notion of uncertainty given by the number of uncertain options introduced in the previous section. 

Note that, $\prec$ only being a preorder on $\mathbb P_\Omega$, in general, two distributions $p',p\in\mathbb P_\Omega$ are not necessarily comparable, i.e. we can have both $p'\not\prec p$ and $p\not\prec p'$. In Figure \ref{fig:15}, we visualize the regions of all comparable distributions for two exemplary distributions on a three-dimensional decision space ($N=3$), represented on the two-dimensional simplex of probability vectors $p=(p_1,p_2,p_3)$. For example, 
\[
p=\big(\tfrac{1}{2},\tfrac{1}{4},\tfrac{1}{4}\big)\, , \quad p'=\big(\tfrac{2}{5},\tfrac{2}{5},\tfrac{1}{5}\big)
\] 
can not be compared under $\prec$, since $\frac{1}{2}>\frac{2}{5}$, but $\frac{3}{4}<\frac{4}{5}$. 

\textit{Cost functions} can now be generalized to probabilistic decision-making by noting that the property $C(A')< C(A)$ whenever $A\subsetneq A'$ in \eqref{DM_elim} means that $C$ is strictly monotonic with respect to the preorder given by set inclusion. 

\begin{Definition}[Cost functions on $\mathbb P_\Omega$] \label{def:costfunctions} We say that a function $C:\mathbb P_\Omega\to\mathbb R_+$ is a cost function, if it is strictly monotonically increasing with respect to the preorder $\prec$, i.e. if 
\begin{equation} \label{schurconvexity}
p'\prec p \quad \Rightarrow \quad C(p') \leq C(p) \, ,
\end{equation}

\noindent with equality only if $p$ and $p'$ are equivalent, $p'\sim p$, which is defined as $p'\prec p$ and $p\prec p'$. Moreover, for a parametrized family of posteriors $(p_r)_{r\in I}$, we say that $r$ is a \textit{resource parameter with respect to a cost function $C$}, if the mapping $I\mapsto \mathbb R_+, r\mapsto C(p_r)$ is strictly monotonically increasing.
\end{Definition}

\noindent Since monotonic functions with respect to majorization were first studied by Schur \cite{Schur1923}, functions with this property are usually called (strictly) \textit{Schur-convex} \cite[Ch.\,3]{Marshall2011}.

\begin{figure}
\centering
\noindent\makebox[1.0\textwidth]{\includegraphics[width = 1.1\textwidth]{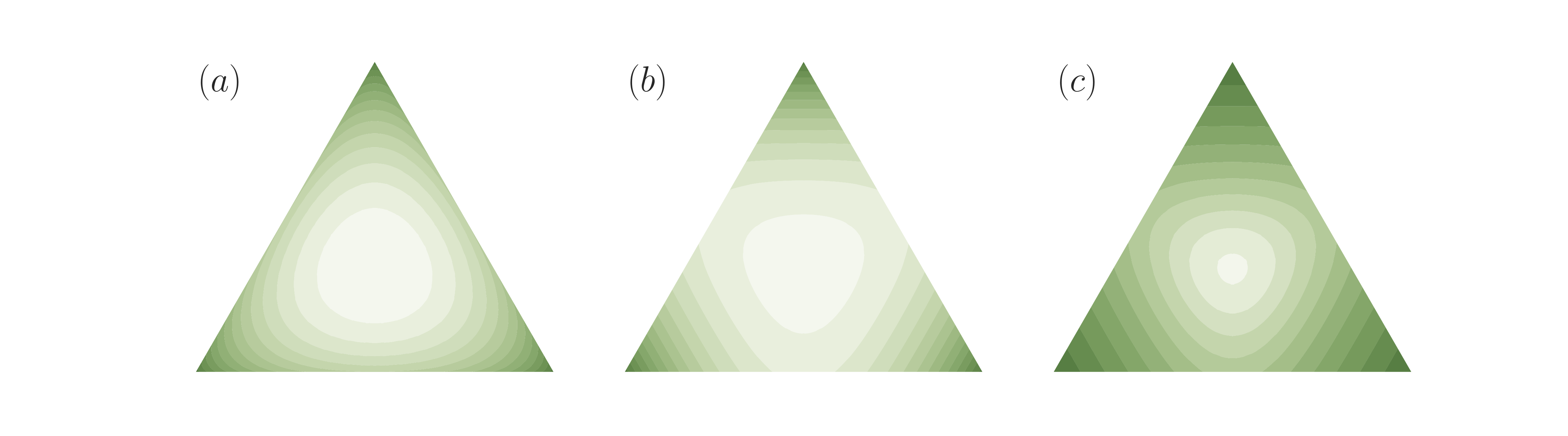}} 

\vspace{-8pt}
\caption{Examples of cost functions for decision spaces with three elements ($N=3$): $(a)$ Shannon entropy, $(b)$ Tsallis entropy of order $\alpha = 4$, $(c)$ R\'{e}nyi entropy of order $\alpha=-3.5$.}
\label{fig:12}
\end{figure}

\begin{Example}[Generalized entropies] \label{ex:genEnt}
{\normalfont
From $(v)$ in Theorem \ref{thm:characterization} it follows that functions of the form
\begin{equation}\label{genEnt}
C(p) =  \sum_{i=1}^N f(p_i) ,
\end{equation}

\noindent where $f$ is strictly convex, are examples of cost functions. Since many entropy measures used in the literature can be seen to be special cases of \eqref{genEnt} (with a concave $f$), functions of this form are often called \textit{generalized entropies} \cite{Canosa2002}. In particular, for the choice $f(t)=t\log t$, we have $C(p) = -H(p)$, where $H(p)$ denotes the \textit{Shannon entropy} of $p$. Thus, if $p'$ contains more uncertainty than $p$ in the sense of Definition \ref{def:uncertainty} ($p'\prec p$) then the Shannon entropy of $p'$ is larger than the Shannon entropy of $p$ and therefore $p'$ contains also more uncertainty in the sense of classical information theory than $p$. Similarly, for $f(t)= -\log(t)$ we obtain the (negative) \textit{Burg entropy}, and for functions of the form $f(t)=\pm t^\alpha$ for $\alpha\in\mathbb R {\setminus} \{0,1\}$ we get the (negative) \textit{Tsallis entropy}, where the sign is chosen depending on $\alpha$ such that $f$ is convex (see e.g. \cite{Gorban2010} for more examples). Moreover, the composition of any (strictly) monotonically increasing function $g$ with \eqref{genEnt} generates another class of cost functions, which contains for example the (negative) \textit{R\'{e}nyi entropy} \cite{Renyi1961}. Note also that entropies of the form \eqref{genEnt} are special cases of \textit{Csisz\'{a}r's \textit{f}-divergences} \cite{Csiszar1972} for uniform reference distributions (see Example \ref{ex:genDiv} below). In Figure \ref{fig:12}, several examples of cost functions are shown for $N=3$. In this case, the 2-dimensional probability simplex $\mathbb P_\Omega$ is given by the triangle in $\mathbb R^3$ with edges $(1,0,0)$, $(0,1,0)$, and $(0,0,1)$. Cost functions are visualized in terms of their level sets.

We prove in Proposition \ref{prop:superadditivity} in the appendix that \textit{all cost functions of the form \eqref{genEnt} are superadditive with respect to coarse-graining}. This seems to be a new result and an improvement upon the fact that generalized entropies (and $f$-divergences) satisfy \textit{information monotonicity} \cite{Amari2009}. More precisely, if a decision in $\Omega$, represented by a random variable $Z$, is split up into two steps by partitioning $\Omega=\bigcup_{i\in I} A_i$ and first deciding about the partition $i\in I$, correspondingly described by a random variable $X$ with values in $I$, and then choosing an option inside of the selected partition $A_i$, represented by a random variable $Y$, i.e. $Z=(X,Y)$, then 
\begin{equation}\label{superadditivity}
C(Z) \geq C(X) + C(Y|X) \, ,
\end{equation}

\noindent where $C(X)\coloneqq C(p(X))$ and $C(Y|X)\coloneqq \mathbb E_{p(X)}[C(p(Y|X))]$. For symmetric cost functions (such as \eqref{genEnt}) this is equivalent to
\begin{equation}\label{superadditivitysymmetric}
C(p_1,\dots,p_N) \geq C(p_1+p_2,p_3,\dots,p_N) + (p_1{+}p_2)\, C(\tfrac{p_1}{p_1+p_2},\tfrac{p_2}{p_1+p_2}) \, .
\end{equation}

\noindent The case of equality in \eqref{superadditivity} and \eqref{superadditivitysymmetric} (see Figure \ref{fig:4a}) is sometimes called  \textit{separability} \cite{Khinchin1957}, \textit{strong additivity} \cite{Csiszar2008}, or \textit{recursivity} \cite{Aczel1974}, and it is often used to characterize Shannon entropy \cite{Faddeev1956,Tverberg1958,Kendall1964,Lee1964,Renyi1961,Aczel1969}. In fact, we also show in the appendix (Proposition \ref{prop:charshannon}) that \textit{cost functions $C$ that are additive under coarse-graining are proportional to the negative Shannon entropy $-H$}. See also Example \ref{ex:genDiv} in the next section, where we discuss the generalization to arbitrary reference distributions. 
}
\end{Example}

We can now refine the notion of a decision-making process introduced in the previous section as a mapping $\phi$ together with a cost function $C$ satisfying \eqref{DM_elim}. Instead of simply mapping from sets $A'$ to smaller subsets $A\subsetneq A'$ by successively eliminating options, we now allow $\phi$ to be a mapping between probability distributions such that $\phi(p)$ can be obtained from $p$ by a finite number of elementary computations (without permutations), and we require $C$ to be a cost function on $\mathbb P_\Omega$, so that
\begin{equation}\label{DM_prob}
p\precnsim \phi(p), \quad C(p) < C(\phi(p)) \qquad \forall p\in\mathbb P_\Omega \, .
\end{equation}

\begin{figure}
\centering
\includegraphics[width=\textwidth]{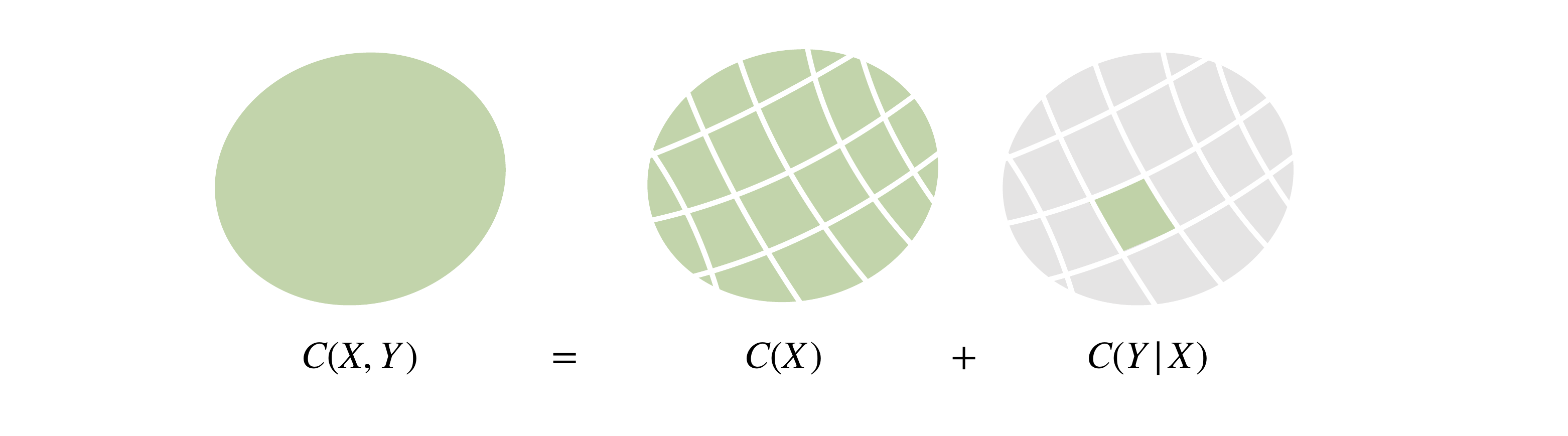}

\vspace{-5pt}
\caption{Additivity under coarse-graining. If the cost for $Z=(X,Y)$ is the sum of the costs for $X$ and the cost for $Y$ given $X$, then the cost function is proportional to Shannon entropy.}
\label{fig:4a}
\end{figure}

\noindent Here, $C(p)$ quantifies the total costs of arriving at a distribution $p$, and $p'\precnsim p$ means that $p'\prec p$ and $p\not\prec p'$. In other words, a decision-making process can be viewed as traversing probability space by moving pieces of probability from one option to another option such that uncertainty is reduced.

Up to now we have ignored one important property of a decision-making process, the distribution $q$ with \textit{minimal cost}, i.e.~satisfying $C(q)\leq C(p)$ for all $p$, which must be identified with the initial distribution of a decision-making process with cost function $C$. As one might expect (see Figure \ref{fig:12}), it turns out that all cost functions according to Definition \ref{def:costfunctions} have the same minimal element.


\begin{Proposition}[Uniform distributions are minimal] \label{prop:unifminimal_main} The uniform distribution $(\tfrac{1}{N},\dots,\tfrac{1}{N})$ is the unique minimal element in $\mathbb P_\Omega$ with respect to $\prec$, i.e.

\vspace{-5pt}
\begin{equation}\label{unifminimal_main}
\big(\tfrac{1}{N},\dots,\tfrac{1}{N} \big) \, \prec \, p \qquad \forall p\in\mathbb P_\Omega \, .
\end{equation}

\end{Proposition}

\noindent Once \eqref{unifminimal_main} is established, it follows from \eqref{schurconvexity} that $C((\tfrac{1}{N},\dots,\tfrac{1}{N})) \leq C(p)$ for all $p$, in particular the uniform distribution corresponds to the initial state of all decision-making processes with cost function $C$ satisfying \eqref{DM_prob}. In particular, it contains the maximum amount of uncertainty with respect to any entropy measure of the form \eqref{genEnt}, known as the second Khinchin axiom  \cite{Khinchin1957}, e.g.~for Shannon entropy $0\leq H(p)\leq \log N$. Proposition \ref{prop:unifminimal_main} follows from characterization $(iv)$ in Theorem \ref{thm:characterization} after noticing that every $p\in\mathbb P_\Omega$ can be transformed to a uniform distribution by permuting its elements cyclically (see Proposition \ref{prop:unifminimal} in the appendix for a detailed proof).

Regarding the possibility that a decision-maker may have \textit{prior information}, for example originating from the experience of previous comparable decision-making tasks, the assumption of a uniform initial distribution seems to be artificial. Therefore, in the following section we arrive at the final notion of a decision-making process by extending the results of this section to allow for arbitrary initial distributions.

\bigskip

\subsection{Decision-making with prior knowledge} \label{sec:probDMprior}

From the discussion at the end of the previous section we conclude that, in full generality, a decision-maker transitions from an initial probability distribution $q\in\mathbb P_\Omega$, called \textit{prior}, to a terminal distribution $p\in \mathbb P_\Omega$, called \textit{posterior}. Note that, since once eliminated options are excluded from the rest of the decision-making process, a posterior $p$ must be \textit{absolutely continuous} with respect to the prior $q$, denoted by $p\ll q$, i.e. $p(x)$ can be non-zero for a given $x\in \Omega$ only if $q(x)$ is non-zero. 

The notion of uncertainty (Definition \ref{def:uncertainty}) can be generalized with respect to a non-uniform prior $q\in\mathbb P_\Omega$ by viewing the probabilities $q_i$ as the probabilities $Q(A_i)$ of partitions $A_i$ of an underlying elementary probability space $\tilde \Omega = \bigcup_i A_i$ of equally likely elements under $Q$, in particular $Q$ represents $q$ as the uniform distribution on $\tilde \Omega$ (see Figure \ref{fig:4b}). The similarity of the entries of the corresponding representation $P\in\mathbb P_{\tilde \Omega}$ of any $p \in \mathbb P_\Omega$ (its uncertainty) then contains information about how close $p$ is to $q$, which we call the \textit{relative uncertainty} of $p$ with respect to $q$ (Definition \ref{def:reluncertainty} below).

The formal construction is as follows: Let $p,q\in\mathbb P_\Omega$ be such that $p\ll q$ and $q_i\in \mathbb Q$. The case when $q_i\in\mathbb R$ then follows from a simple approximation of each entry by a rational number. Let $\alpha\in \mathbb N$ be such that $\alpha \, q_i \in \mathbb N$ for all $i\in\{1,\dots,N\}$, for example $\alpha$ could be chosen as the least common multiple of the denominators of the $q_i$. The underlying elementary probability space $\tilde \Omega$ then consists of $\alpha$ elements and there exists a partitioning $\{A_i\}_{i=1,\dots,N}$ of $\tilde \Omega$ such that 
\begin{equation} \label{def:partitionsize}
|A_i| = \alpha\,q_i \qquad \forall i\in\{1,\dots, N\} \, ,
\end{equation}

\noindent where $Q$ denotes the uniform distribution on $\tilde \Omega$. In particular, it follows that 
\begin{equation}\label{Qunif}
Q(A_i) = \sum_{j=1}^{|A_i|}  \frac{1}{\alpha}  = q_i  \qquad  \forall i\in\{1,\dots,N\} \, , 
\end{equation}

\noindent i.e.~$Q$ represents $q$ in $\tilde \Omega$ with respect to the partitioning $\{A_i\}_{i}$. Similarly, any $p\in\mathbb P_\Omega$ can be represented as a distribution on $\tilde \Omega$ by requiring that $P(A_i) = p_i$ for all $i\in\{1,\dots,N\}$ and letting $P$ to be constant inside of each partition, i.e.~similar to \eqref{Qunif} we have $P(A_i) = |A_i| \, P(\omega) = p_i$ for all $\omega\in A_i$ and therefore by \eqref{def:partitionsize}
\begin{equation} \label{def:qrep}
P(\omega) \ = \ \frac{1}{\alpha} \, \frac{p_i}{q_i} \qquad \forall \omega \in A_i \ .
\end{equation}

\noindent Note that, if $q_i=0$ then $p_i=0$ by absolute continuity ($p\ll q$) in which case we can either exclude option $i$ from $\Omega$ or set $P(\omega)=0$.

\begin{figure}
\centering
\includegraphics[width=.8\textwidth]{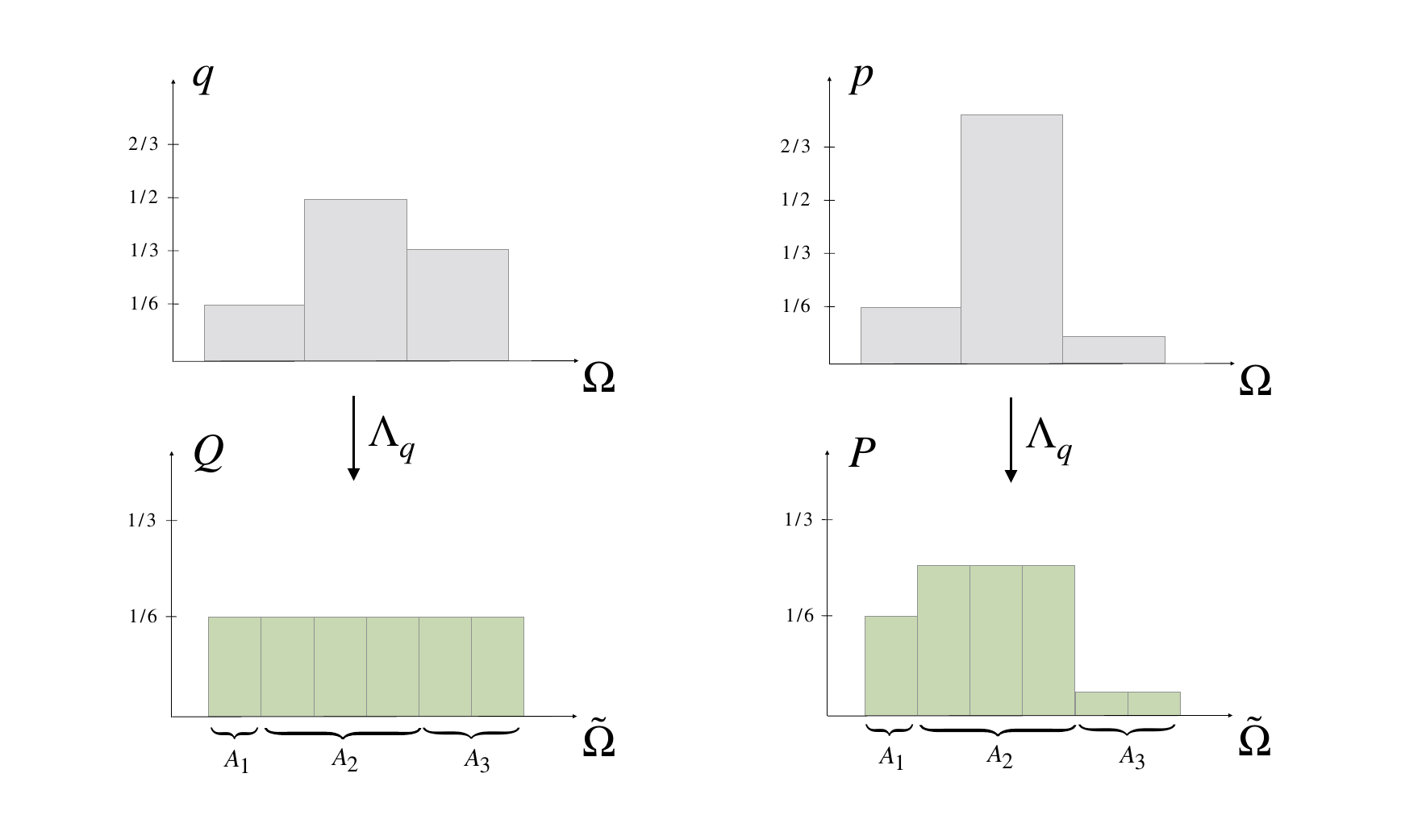}

\vspace{-5pt}
\caption{Representation of $q$ and $p$ by $Q$ and $P$ on $\tilde \Omega$ (Example \ref{ex:unifrep}), such that the probabilities $q_i$ and $p_i$ are given by the probabilities of the partitions $A_i$ with respect to $Q$ and $P$, respectively.}
\label{fig:4b}
\end{figure}   

\begin{Example} \label{ex:unifrep}
{\normalfont

For a prior \smash{$q=(\frac{1}{6},\frac{1}{2},\frac{1}{3})$} we put $\alpha=6$, so that $\tilde \Omega = \{\omega_1,\dots,\omega_6\}$ should be partitioned as $\tilde \Omega = \{ \omega_1 \} \cup \{\omega_2,\omega_3,\omega_4\} \cup \{\omega_5,\omega_{6}\}$. Then $q_i$ corresponds to the probability of the $i$-th partition under the uniform distribution \smash{$Q = \frac{1}{6} (1,\dots,1)$}, while \smash{$p=(\frac{1}{6},\frac{3}{4},\frac{1}{12})$} is represented on $\tilde \Omega$ by \smash{$P = \big(\tfrac{1}{6}, \tfrac{1}{4},\tfrac{1}{4},\tfrac{1}{4}, \tfrac{1}{24}, \tfrac{1}{24})$} (see Figure \ref{fig:4b}).
}
\end{Example}

Importantly, if the components of the representation $\Lambda_q p \coloneqq P$ in $\mathbb P_{\tilde \Omega}$ given by \eqref{def:qrep} are similar to each other, i.e.~if $P$ is close to uniform, then the components of $p$ must be very similar to the components of $q$, which we express by the concept of \textit{relative uncertainty}.

\begin{Definition}[Uncertainty relative to $q$]\label{def:reluncertainty} We say that $p'\in\mathbb P_\Omega$ contains more uncertainty with respect to a prior $q\in\mathbb P_\Omega$ than $p\in\mathbb P_\Omega$, denoted by $p' \prec_q p$, if and only if $\Lambda_q p'$ contains more uncertainty than $\Lambda_q p$, i.e.
\begin{equation}\label{eq:reluncertainty}
p'  \prec_q  p  \quad :\Leftrightarrow \quad \Lambda_q p'  \prec \Lambda_q p 
\end{equation}

\noindent where $\Lambda_q:\mathbb P_\Omega\to \mathbb P_{\tilde \Omega}, p\mapsto P$ is given by \eqref{def:qrep}.
\end{Definition}

As we will show in Theorem \ref{thm:charrel} below, it turns out that $\prec_q$ coincides with a known concept called \textit{$q${-}majorization} \cite{Veinott1971}, \textit{majorization relative to $q$} \cite{Joe1990, Marshall2011}, or \textit{mixing distance} \cite{Ruch1978}. Due to the lack of a characterization by partial sums, it is usually defined as a generalization of characterization $(iii)$ in Theorem \ref{thm:characterization}, that is $p'$ is $q$-majorized by $p$ iff $p' = pA$, where $A$ is a so-called \textit{$q$-stochastic matrix}, which means that it is a stochastic matrix ($Ae = e$) with $qA = q$. In particular, $\prec_q$ does not depend on the choice of $\alpha$ in the definition of $\Lambda_q$. Here, we provide two new characterizations of $q$-majorization, the one given by Definition \ref{def:reluncertainty}, and one using partial sums generalizing the original definition of majorization.

\medskip 

\begin{Theorem}[Characterizations of $p'\prec_q p$]\label{thm:charrel} The following are equivalent
\begin{enumerate}
\item[$(i)$] $p'\prec_q p$, i.e. $p'$ contains more uncertainty relative to $q$ than $p$ (Def. \ref{def:reluncertainty})\\[-4pt]
\item[$(ii)$] $\Lambda_q p$ can be obtained from $\Lambda_q p'$ by a finite number of elementary computations and permutations on $\mathbb P_{\tilde \Omega}$\\[-4pt]
\item[$(iii)$] $p' = pA$ for a $q$-stochastic matrix $A$, i.e.~$Ae=e$ and $qA = q$\\[-4pt]
\item[$(iv)$] $\sum_{i=1}^N q_i f\big(\frac{p'_i}{q_i} \big) \leq \sum_{i=1}^N q_i f\big(\frac{p_i}{q_i}\big)$ for all continuous convex functions $f$\\[-4pt]
\item[$(v)$] $\sum_{i=1}^{l-1} (p'_i)^{\downarrow} + a_q(k,l) (p'_l)^{\downarrow} \ \leq \ \sum_{i=1}^{l-1}p_i^\downarrow + a_q(k,l) p_l^\downarrow$ for all $k,l\in\mathbb N$ that satisfy $\alpha \sum_{i=1}^{l-1} q_i^\downarrow \leq k \leq \alpha \sum_{i=1}^l q_i^\downarrow$ and $1\leq l \leq N$, where the arrows indicate that $(p_i^\downarrow/q_i^\downarrow)_i$ is ordered decreasingly, and $a_q(k,l)\coloneqq (\frac{k}{\alpha} - \sum_{i=1}^{l-1} q_i^\downarrow)/q_l^\downarrow$.
\end{enumerate}
\end{Theorem}

To prove that $(i)$, $(iii)$, and $(v)$ are equivalent (see Proposition \ref{prop:charrel} in the appendix), we use of the fact that $\Lambda_q:\mathbb P_\Omega\to\mathbb P_{\tilde \Omega}$ has a left inverse $\Lambda_q^{-1}: \Lambda_q(\mathbb P_\Omega) \to \mathbb P_{\Omega}$. This can be verified by simply multiplying the corresponding matrices given in the proof of Proposition \ref{prop:charrel}. The equivalence betweeen $(iii)$ and $(iv)$ is shown in \cite{Joe1990} (see also \cite{Ruch1978,Marshall2011}). Characterization $(ii)$ follows immediately from Definition \ref{def:uncertainty} and Definition \ref{def:reluncertainty}.

As required from the discussion at the end of the previous section, $q$ is indeed minimal with respect to $\prec_q$, which means that it contains the most amount of uncertainty with respect to itself.

\begin{Proposition}[The prior is minimal] \label{priorminimal_main}
The prior $q\in\mathbb P_\Omega$ is the unique minimal element in $\mathbb P_\Omega$ with respect to $\prec_q$, that is
\begin{equation}
q\prec_q p \qquad \forall p\in\mathbb P_\Omega \, .
\end{equation}
\end{Proposition}

\noindent This follows more or less directly from Proposition \ref{prop:unifminimal_main} and the equivalence of $(i)$ and $(iii)$ in Theorem \ref{thm:charrel} (see Proposition \ref{priorminimal} in the appendix for a detailed proof).

Order-preserving functions with respect to $\prec_q$ generalize cost functions introduced in the previous section (Definition \ref{def:costfunctions}). According to Proposition \ref{priorminimal_main}, such functions have a unique minimum given by the prior $q$. Since cost functions are used in Definition \ref{def:DM} below to quantify the expenses of a decision-making process, we require their minimum to be zero, which can always be achieved by redefining a given cost function by an additive constant.

\begin{Definition}[Cost functions relative to $q$] \label{def:relcostfunctions} We say that a function $C_q:\mathbb P_\Omega\to\mathbb R_+$ is a cost function relative to $q$, if $C_q(q) = 0$, if it is invariant under relabeling $(q_i,p_i)_i$, and if it is strictly monotonically increasing with respect to the preorder $\prec_q$, that is if 
\begin{equation} \label{relschurconvexity}
p'\prec_q p \quad \Rightarrow \quad C_q(p') \leq C_q(p) \, ,
\end{equation}

\noindent with equality only if $p'\sim_q p$, i.e.~if $p'\prec_q p$  and $p\prec_q p'$. Moreover, for a parametrized family of posteriors $(p_r)_{r\in I}$, we say that $r$ is a \textit{resource parameter with respect to a cost function $C_q$}, if the mapping $I\mapsto \mathbb R_+, r\mapsto C_q(p_r)$ is strictly monotonically increasing.
\end{Definition}

\begin{figure}
\centering
\noindent\makebox[1.0\textwidth]{\includegraphics[width = 1.1\textwidth]{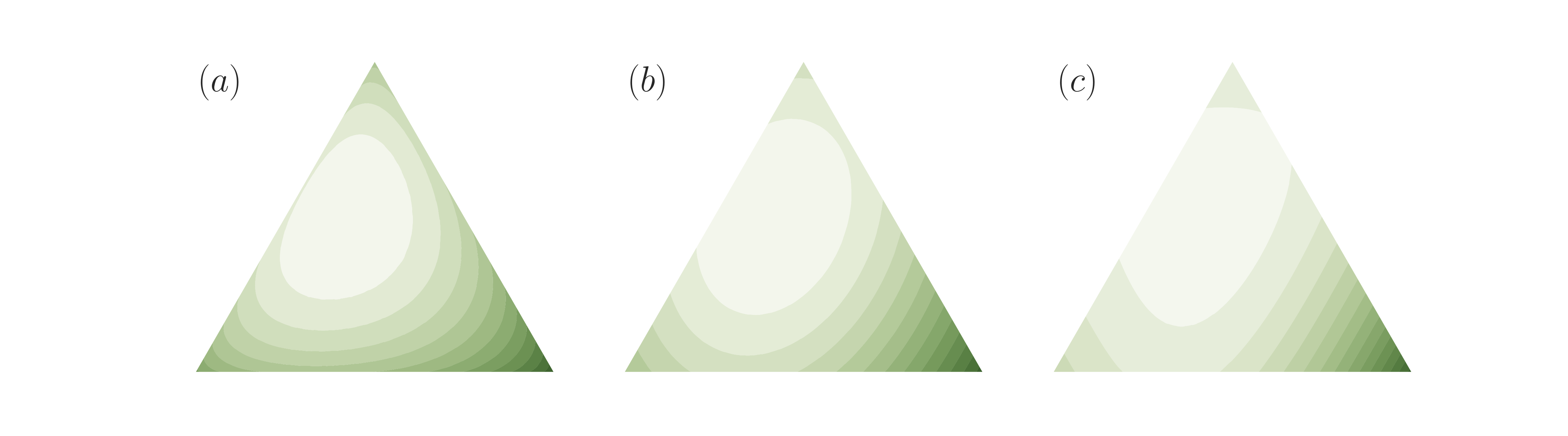}} 

\vspace{-5pt}
\caption{Examples of cost functions for $N=3$ relative to $q=(\frac{1}{3},\frac{1}{2},\frac{1}{6})$. $(a)$ Kullback-Leibler divergence, $(b)$ Squared $\ell^2$ distance, $(c)$ Tsallis relative entropy of order $\alpha=3.0$.}
\label{fig:13}
\end{figure}

Similar to generalized entropy functions discussed in Example \ref{ex:genEnt}, in the literature there are many examples of relative cost functions, usually called \textit{divergences} or \textit{measures of divergence}. 
\begin{Example}[$f$-divergences] \label{ex:genDiv}
{\normalfont
From $(iv)$ in Theorem \ref{thm:charrel} it follows that functions of the form
\begin{equation}\label{genDiv}
C_q(p) \coloneqq \sum_{i=1}^N q_i \, f\Big(\frac{p_i}{q_i} \Big) ,
\end{equation}

\noindent where $f$ is continuous and strictly convex with $f(1)=0$, are examples of cost functions relative to $q$. Many well-known divergence measures can be seen to belong to this class of relative cost functions, also known as \textit{Csisz\'{a}r's $f$-divergences} \cite{Csiszar1972}: the \textit{Kullback-Leibler divergence} (or \textit{relative entropy}), the squared \textit{$\ell^2$ distance}, the \textit{Hartley entropy}, the \textit{Burg entropy}, the \textit{Tsallis entropy}, and many more \cite{Gorban2010,Csiszar2008} (see Figure \ref{fig:13} for visualizations of some of them in $N=3$ relative to a non-uniform prior).

As a generalizition of Proposition \ref{prop:superadditivity} (superadditivity of generalized entropies), we prove in Proposition \ref{prop:superadditivityDiv} in the appendix that \textit{$f$-divergences are superadditive under coarse-graining}, that is, for $Z=(X,Y)$
\begin{equation} \label{relsuperadditivity}
C_q(Z) \geq C_q(X) + C_q(Y|X)
\end{equation}

\noindent whenever $C_q(X)\coloneqq C_{q(X)}(p(X))$ and $C_q(Y|X)\coloneqq \mathbb E_{p(X)}[C_{q(Y|X)}(p(Y|X))]$. This generalizes \eqref{superadditivity} to the case of a non-uniform prior. Similar to entropies, the case of equality in \eqref{relsuperadditivity} is sometimes called \textit{composition rule} \cite{Hobson1969}, \textit{chain rule} \cite{Leinster2017}, or \textit{recursivity} \cite{Csiszar2008}, and is often used to characterize Kullback-Leibler divergence \cite{Hobson1969,Mattsson2002,Csiszar2008,Leinster2017}. 

Indeed, we also show in the appendix (Proposition \ref{prop:charDKL}) that \textit{all additive cost functions with respect to $q$ are proportional to Kullback-Leibler divergence (relative entropy)}. This goes back to Hobson's modification \cite{Hobson1969} of Shannon's original proof \cite{Shannon1948}, after establishing the following monotonicity property for uniform distributions: If $f(M,N)$ denotes the cost $C_{u_N}(u_M)$ of a uniform distribution $u_M$ over $M$ elements relative to a uniform distribution $u_N$ over $N\geq M$ elements, then (see Figure \ref{fig:3a})
\begin{equation} \label{uniformmonotonicity_main}
\begin{array}{rl}
f(M',N) \leq f(M,N) &\ \forall M \leq M' \leq N \, ,\\[5pt]
f(M,N) \geq f(M,N') &\  \forall M \leq N' \leq N \, . 
\end{array}
\end{equation}

\noindent Note that, even though our proof of Proposition \ref{prop:charDKL} uses additivity under coarse graining to show the monotonicity property \eqref{uniformmonotonicity_main}, it is easy to see that any relative cost function of the form \eqref{genDiv} also satisfies \eqref{uniformmonotonicity_main} by using the convexity of $f$ in the form $f(t)\leq \tfrac{t}{s} f(s)  + (1-\tfrac{t}{s}) f(0)$ with $t = \frac{N'}{M} < \frac{N}{M} = s$.

\begin{figure}
\centering
\includegraphics[width=1.0\textwidth]{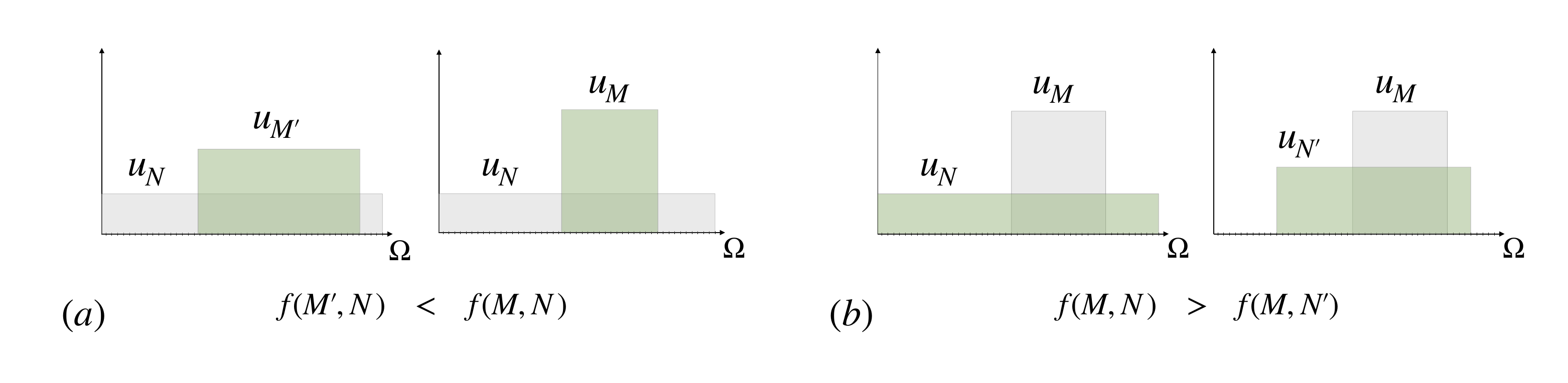}

\vspace{-5pt}
\caption{Monotonicity property \eqref{uniformmonotonicity_main}. $(a)$ The cost is higher when more uncertainty has been reduced. $(b)$ If the posterior is the same, then it is cheaper to start from a prior with fewer options.}
\label{fig:3a}
\end{figure}   

In terms of decision-making, superadditivity under coarse-graining means that decision-making costs can potentially be reduced by splitting up the decision into multiple steps, for example by a more intelligent search strategy. For example, if $N=2^k$ for some $k\in\mathbb N$ and $C_q$ is superadditive, then the cost for reducing uncertainty to a single option, i.e.~$p=(1,0,\dots,0)$, when starting from a uniform distribution $q$, satisfies
\[
C_q(p) \ \geq C_{q^2}(1,0) + C_{q^{N/2}}(1,0,\dots,0) \ \geq \, \dots \, \geq \  \log N \, = \, D_{KL}(p\| q) \, ,
\]

\noindent where $q^{n} \coloneqq (\frac{1}{n},\dots,\frac{1}{n})$, and we have set \smash{$C_{q^2}(1,0) = 1$} as unit cost (corresponding to \textit{1 bit} in the case of Kullback-Leibler divergence). Thus, intuitively the property of the Kullback-Leibler divergence of being additive under coarse-graining might be viewed as describing the minimal amount of processing costs that must be contained in any cost function, because \textit{it cannot be reduced by changing the decision-making process}. Therefore, in the following we call cost functions that are proportional to the Kullback-Leibler divergence simply \textit{informational costs}.
}
\end{Example}

In contrast to the previous section, in the definition of $\prec_q$ and its characterizations we have never used elementary computations on $\mathbb P_\Omega$ directly. This is due to the fact that permutations do interact with the uncertainty relative to $q$, and therefore $\prec_q$ cannot be characterized by a finite number of elementary computations and permutations on $\mathbb P_\Omega$. However, we can still define elementary computations relative to $q$ by the inverse of Pigou-Dalton transfers $T_\varepsilon$ of the form \eqref{eq:pigoudalton} such that $T_\varepsilon p \precnsim_q p$ for $\varepsilon >0$, which is arguably the most basic form of how to generate uncertainty with respect to $q$. 

Even for small $\varepsilon$, a regular Pigou-Dalton transfer does not necessarily increase uncertainty relative to $q$, because the similarity of the components now needs to be considered with respect to $q$. Instead, we compare the components of the representation $P=\Lambda_q p$ of $p\in\mathbb P_\Omega$, and move some probability weight $\varepsilon \geq 0$ from $P(A_n)$ to $P(A_m)$ whenever $P(\omega)\leq P(\omega')$ for $\omega\in A_m$ and $\omega'\in A_n$, by distributing $\varepsilon$ evenly among the elements in $A_m$ (see Figure \ref{fig:4c}), denoted by the transformation $\tilde T_\varepsilon$. Here, $\varepsilon$ must be small enough such that the inequality \smash{$\frac{1}{\alpha} \frac{p_m}{q_m} = P(\omega) \leq P(\omega') = \frac{1}{\alpha} \frac{p_n}{q_n}$} is invariant under $\tilde T_\varepsilon$, which means that
\begin{equation}\nonumber
(\tilde T_\varepsilon P)(\omega) \leq (\tilde T_\varepsilon P)(\omega') \quad \Leftrightarrow \quad \frac{1}{\alpha} \frac{p_m}{q_m} + \frac{\varepsilon}{|A_m|} \ \leq  \ \frac{1}{\alpha}\frac{p_n}{q_n} - \frac{\varepsilon}{|A_n|}
\end{equation}
and therefore 
\begin{equation}\label{relativeepsilonbound}
\varepsilon \ \leq \ \frac{\tfrac{p_n}{q_n}-\frac{p_m}{q_m}}{\frac{1}{q_m}+\frac{1}{q_n}}.
\end{equation}

\noindent By construction, $\tilde T_\varepsilon$ minimally increases uncertainty in $\mathbb P_{\tilde \Omega}$ while staying in the image of $\mathbb P_\Omega$ under $\Lambda_q$, by keeping the values of $P$ constant in each partition, and therefore $T_\varepsilon \coloneqq \Lambda_q^{-1} \tilde T_\varepsilon \Lambda_q$ can be considered as the most basic way of how to increase uncertainty relative to $q$.

\begin{figure}
\centering
\includegraphics[width=.9\textwidth]{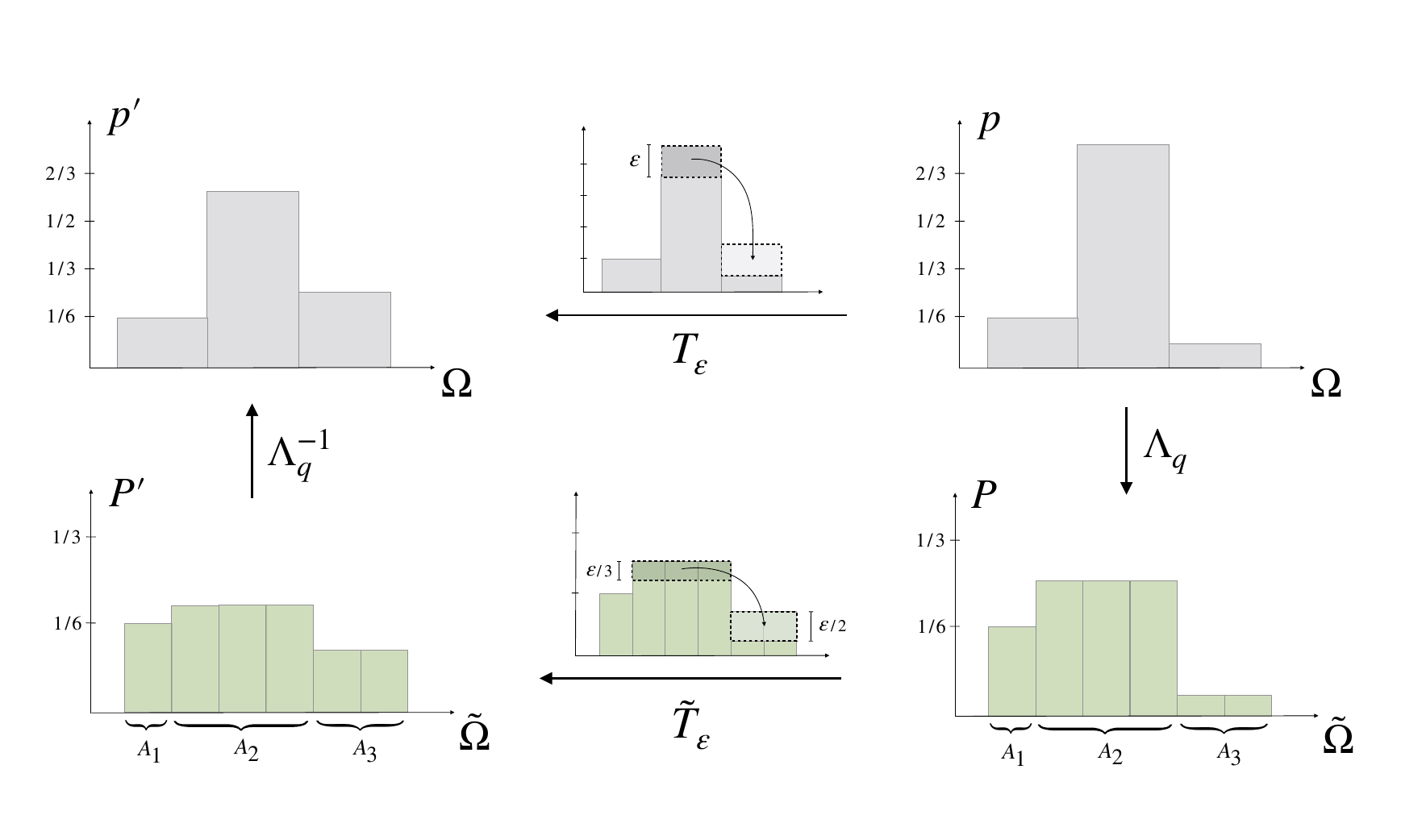}

\vspace{-6pt}
\caption{Pigou-Dalton transfer relative to $q$. A distribution $p\in\mathbb P_\Omega$ is transformed relative to $q$ by first moving some amount of weight $\varepsilon\geq 0 $ from $P(A_n)$ to $P(A_m)$ where $n,m$ are such that $P(\omega)\leq P(\omega')$ whenever $\omega\in A_m$ and $\omega'\in A_n$, with $\varepsilon$ small enough such that this relation remains true after the transformation, and then mapping the transformed distribution back to $\mathbb P_\Omega$ by $\Lambda_q^{-1}$ (see Definition \ref{def:relpigoudalton}).}
\label{fig:4c}
\end{figure}

\begin{Definition}[Elementary computation relative to $q$] \label{def:relpigoudalton}
We call a transformation on $\mathbb P_\Omega$ of the form 
\begin{equation}\label{eq:relpigoudalton}
T_\varepsilon: p\mapsto (p_1, \dots,  p_m+\varepsilon, \dots, p_{n} - \varepsilon, \dots, p_N) \, ,
\end{equation}

\noindent with $m,n$ such that \smash{$\frac{p_m}{q_m}\leq \frac{p_n}{q_n}$}, and $\varepsilon$ satisfying \eqref{relativeepsilonbound}, a Pigou-Dalton transfer relative to $q$, and its inverse an elementary computation relative to $q$.

\end{Definition}

We are now in the position to state our final definition of a decision-making process.

\begin{Definition}[Decision-making process] \label{def:DM}
A decision-making process is a gradual transformation 
\[
q \longrightarrow \cdots \longrightarrow p' \longrightarrow \phi(p')\longrightarrow \cdots \longrightarrow p
\]
of a prior $q\in\mathbb P_\Omega$ to a posterior $p\in\mathbb P_\Omega$, such that each step decreases uncertainty relative to $q$. This means that $p$ is obtained from $q$ by successive application of a mapping $\phi$ between probability distributions on $\Omega$, such that $\phi(p')$ can be obtained from $p'$ by finitely many elementary computations relative to $q$, in particular
\begin{equation}\label{DM}
q\precnsim_q  p' \precnsim_q \phi(p'),  \qquad 0 = C_q(q)< C_q(p') < C_q(\phi(p')) \, ,
\end{equation}

\noindent where $C_q(p')$ quantifies the total costs of a distribution $p'$, and $p' \precnsim_q p$ means that $p'\prec_q p$ and $p\not\prec_q p'$.  
\end{Definition}

In other words, a decision-making process can be viewed as traversing probability space from prior $q$ to posterior $p$ by moving pieces of probability from one option to another option such that uncertainty is reduced relative to $q$, while expending a certain amount of resources determined by the cost function $C_q$.

\bigskip

\section{Bounded rationality}\label{sec:bounded} 

\subsection{Bounded rational decision-making}

In this section, we consider decision-making processes that trade off utility against costs. Such decision-makers either maximize a utility function subject to a constraint on the cost function, for example an author of a scientific article that optimizes the article's quality until a deadline is reached, or minimizing the cost function subject to a utility constraint, for example a high-school student that minimizes effort such that the requirement to pass a certain class is achieved. In both cases, the decision-makers are called \textit{bounded rational}, since in the limit of no resource constraints they coincide with \textit{rational} decision-makers. 

In general, depending on the underlying system, such an optimization process might have additional \textit{process dependent} constraints that are not directly given by resource limitations, for example in cases when the optimization takes place in a parameter space that has less degrees of freedom than the full probability space $\mathbb P_\Omega$. Abstractly, this is expressed by allowing the optimization process to search only in a subset $\Gamma\subset \mathbb P_\Omega$.

\begin{Definition}[Bounded rational decision-making process]\label{def:BRDM}
Let $U:\Omega \to\mathbb R$ be a given utility function, and $\Gamma\subset \mathbb P_\Omega$. A decision-making process with prior $q$, posterior $p^\ast \in \Gamma$, and cost function $C_q$ is called bounded rational if its posterior satisfies
\begin{equation} \label{eq:optU}
p^\ast = \mathop{\mathrm{argmax}}_{p\in \Gamma} \Big \{ \mathbb E_p[U]  \, \Big| \, C_q(p) \leq C_0 \Big \} \, ,
\end{equation}

\noindent for a given upper bound $C_0\geq 0$, or equivalently
\begin{equation}\label{eq:optC}
p^\ast = \mathop{\mathrm{argmin}}_{p\in \Gamma} \Big \{ C_q(p) \, \Big| \, \mathbb E_p[U] \geq U_0 \Big \} \, ,
\end{equation}

\noindent for a given lower bound $U_0\in\mathbb R$. In the case when the process constraints disappear, i.e.~if $\Gamma=\mathbb P_\Omega$, then a bounded rational decision-maker is called \textit{bounded-optimal}.
\end{Definition}

The equivalence between \eqref{eq:optU} and \eqref{eq:optC} is easily seen from the equivalent optimization problem given by the formalism of Lagrange multipliers \cite{Everett1963},
\begin{equation}\label{eq:unconstrained}
p_\beta \coloneqq  \mathop{\mathrm{argmin}}_{p\in \Gamma} \Big ( C_q(p) - \beta\,  \mathbb E_p[U] \Big) = \mathop{\mathrm{argmax}}_{p\in \Gamma} \Big ( \mathbb E_p[U] - \tfrac{1}{\beta}\,  C_q(p) \Big) \, ,
\end{equation}

\noindent where the cost or utility constraint is expressed by a trade-off between utility and cost, or cost and utility, with a trade-off parameter given by the Lagrange multiplier $\beta$, which is chosen such that the constraint given by $C_0$ or $U_0$ is satisfied. It is easily seen from the maximization problem on the right side of \eqref{eq:unconstrained} that a larger value of $\beta$ decreases the weight of the cost term and thus allows for higher values of the cost function. Hence, $\beta$ parametrizes the amount of resources the decision-maker can afford with respect to the cost function $C_q$, and, at least in non-trivial cases (non-constant utilities) it is therefore a resource parameter with respect to $C_q$ in the sense of Definition \ref{def:relcostfunctions}. In particular, for $\beta\to 0$, the decision-maker minimizes its cost function irrespective of the expected utility, and therefore stays at the prior, $p_0 = q$, whereas $\beta\to\infty$ makes the cost function disappear so that the decision-maker becomes purely rational with a Dirac posterior centered on the optima $x^\ast$ of the utility function $U$.

For example, in Figure \ref{fig:14} we can see how the posteriors $(p_\beta)_{\beta\geq 0}$ of bounded-optimal decision-makers with different cost functions for $N=3$ and with utility $U=(0.8,1.0,0.4)$ leave a trace in probability space, by moving away from an exemplary prior $q=(\frac{1}{3},\frac{1}{2},\frac{1}{6})$ and eventually arriving at the rational solution $\delta_{(0,1,0)}$.

\begin{figure}
\centering
\noindent\makebox[1.0\textwidth]{\includegraphics[width = 1.1\textwidth]{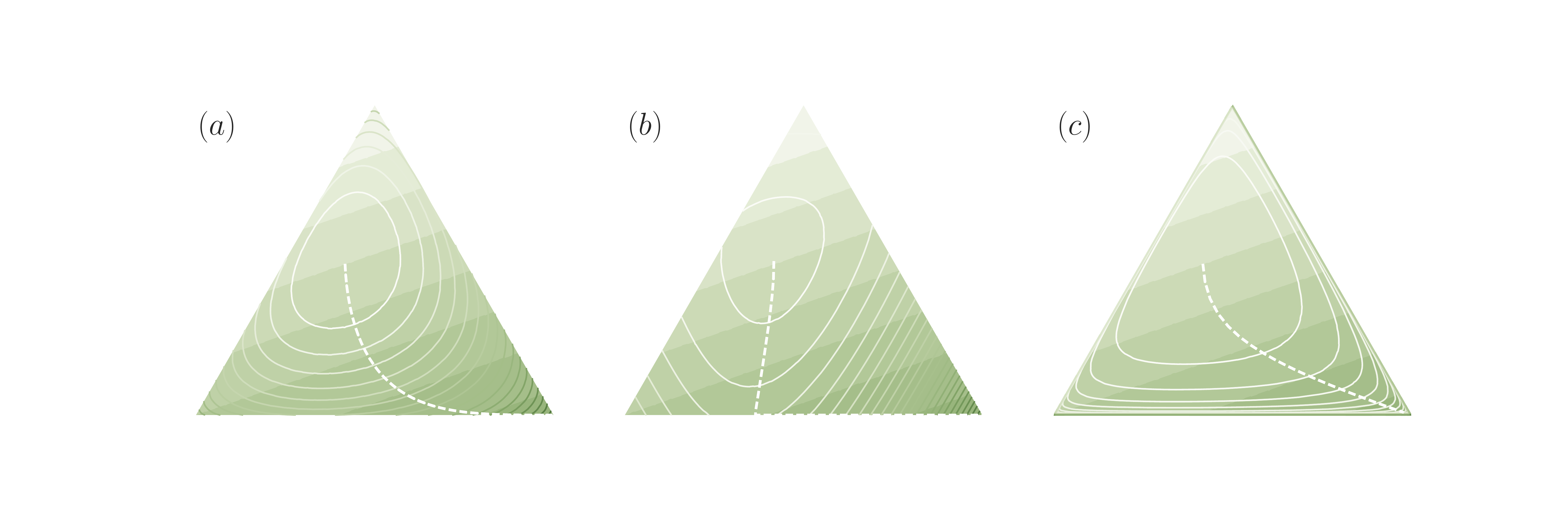}} 

\vspace{-20pt}
\caption{Paths of bounded-optimal decision-makers in $\mathcal P(\Omega)$ for $N=3$. The straight lines in the background denote level sets of expected utility, the solid lines are level sets of the cost functions, and the dashed curves represent the paths $(p_\beta)_{\beta\geq 0}$ of a bounded-optimal decision-maker given by \eqref{eq:unconstrained} with $U=(0.8,1.0,0.4)$, prior $q=(\frac{1}{3},\frac{1}{2},\frac{1}{6})$, and cost functions given by $(a)$ Kullback-Leibler divergence, $(b)$ Tsallis relative entropy of order $\alpha = 3$, and $(c)$ Burg relative entropy.}
\label{fig:14}
\end{figure}

For informational costs (i.e.~proportional to Kullback-Leibler divergence), $\beta$ is a resource parameter with respect to \textit{any} cost function. 

\begin{Proposition}\label{brdmproc}
If $(p_\beta)_{\beta\geq 0}$ is a family of bounded-optimal posteriors given by \eqref{eq:unconstrained} with $C_q(p) = D_{KL}(p\|q)$, then $\beta$ is a resource parameter with respect to any cost function, in particular
\begin{equation}\label{BRprec}
q = p_0 \, \precnsim_q \, p_{\beta'} \, \precnsim_q \, p_{\beta} \qquad \forall \beta',\beta \ \ \text{with} \ \ \beta'<\beta \,.
\end{equation}
\end{Proposition}

This generalizes a result in \cite{Rossignoli2004} to the case of non-uniform priors, by making use of our new characterization $(v)$ of $\prec_q$, by which it suffices to show that the function \smash{$\beta \mapsto \sum_{i=1}^{l-1} (p_{\beta,i})^{\downarrow} + a_q(k,l) (p_{\beta,l})^{\downarrow}$} is monotonically increasing for all $k,l$ specified in Theorem \ref{thm:charrel} (see Proposition \ref{prop:boundedoptDM} in the appendix for details). For simplicity, we restrict ourselves to the case of the Kullback-Leibler divergence, however the proof is analogous for cost functions of the form \eqref{genDiv} with $f$ being differentiable and strictly convex on $[0,1]$ (so that $f'$ is strictly monotonically increasing and thus invertible on $[0,1]$, see \cite{Rossignoli2004} for the case of uniform priors).

Hence, for any $\beta >0$, the posteriors $(p_{\beta'})_{\beta'<\beta}$ of a bounded-optimal decision-making process with the Kullback-Leibler divergence as cost function can be regarded as the steps of a decision-making process (i.e.~satisfying \eqref{DM}) with posterior $p_\beta$, where each step optimally trades off utility against informational cost. This means that with increasing $\beta$ the posteriors $p_\beta$ do not only decrease entropy in the sense of the Kullback-Leibler divergence, but also in the sense of any other cost function.

The important case of bounded-optimal decision-makers with informational costs is termed \textit{information-theoretic bounded rationality} \cite{Ortega2010,Tishby2011,Ortega2013} and is studied more closely in the following sections.

\bigskip

\subsection{Information-theoretic bounded rationality} \label{sec:inf-theo-br}

For bounded-optimal decision-making processes with informational costs the unconstrained optimization problem \eqref{eq:unconstrained} takes the form $\max_{p\in \mathbb P_\Omega} \mathcal F[p]$, where
\begin{equation} \label{freeenergy}
\mathcal F[p]\coloneqq \mathbb E_p[U] - \tfrac{1}{\beta} D_{KL}(p\|q) \, ,
\end{equation}

\noindent which has a unique maximum $p_\beta$, the bounded-optimal posterior given by
\begin{equation} \label{boltzmann}
p_\beta(x) \, = \, \frac{1}{Z_\beta} \, q(x) \, e^{\beta U(x)} \, 
\end{equation}

\noindent with normalization constant $Z_\beta$. This form can easily be derived by finding the zeros of the functional derivative of the objective functional \eqref{freeenergy} with respect to $p$ (with an additional normalization constraint), whereas the uniqueness follows from the convexity of the mapping $p\mapsto D_{KL}(p\|q)$. For the actual maximum of $\mathcal F$ we obtain 
\[
\mathcal F_\beta \coloneqq \max_{p\in \mathbb P_\Omega} \mathcal F[p] \ = \ \mathcal F[p_\beta] \ = \ \frac{1}{\beta} \log Z_\beta \, ,
\]

\noindent so that $p_\beta(x) = q(x) \, e^{\beta (U(x) - \mathcal F_\beta)}$.

Due to its analogy with physics, in particular thermodynamics (see e.g. \cite{Ortega2013}), the maximization of \eqref{freeenergy} is known as the \textit{Free Energy principle} of \textit{information-theoretic bounded rationality}, pioneered in \cite{Ortega2010,Tishby2011,Ortega2013}, further developed in \cite{Genewein2015,Gottwald2018}, and applied in recent studies of artificial systems, such as generative neural networks trained by Markov chain Monte Carlo methods \cite{Hihn2018}, or in reinforcement learning as an adaptive regularization strategy \cite{Leibfried2018,Moya2019}, as well as in recent experimental studies on human behavior \cite{Ortega2016a,Schach2018}. Note that there is a formal connection of \eqref{freeenergy} and the Free Energy principle of \textit{active inference} \cite{Friston2009}, however, as discussed in \cite[Sec.~6.3]{Gottwald2018}, both Free Energy principles have conceptually different interpretations.

\begin{Example}[Bayes rule as a bounded-optimal posterior]\label{ex:bayesianinference}

{\normalfont
In Bayesian inference, the parameter $\theta$ of the distribution $p_\theta$ of a random variable $Y$ is inferred from a given dataset $d=\{y_1,\dots,y_N\}$ of observations of $Y$ by treating the parameter itself as a random variable $\Theta$ with a prior distribution $q(\Theta)$. The parametrized distribution of $Y$ evaluated at an observation $y_i\in d$ \textit{given} a certain value of $\Theta$, i.e.~$p(y_i|\Theta{=}\theta)$, is then understood as a function of $\theta$, known as the \textit{likelihood} of the datapoint $y_i$ under the assumption of $\Theta=\theta$. After seeing the dataset $d$, the belief about $\Theta$ is updated by using Bayes rule
\[
p(\theta) \, = \, \frac{q(\theta) p(d|\theta)}{\mathbb E_{q(\Theta)}[p(d|\Theta)] } \, .
\]

\noindent This takes the form of a bounded-optimal posterior \eqref{boltzmann} with $\beta = N$ and utility function given by the average log-likelihood per datapoint,
\[
U(\theta) \coloneqq \frac{1}{N} \log p(d|\theta) = \frac{1}{N} \sum_{i=1}^N \log(p(y_i|\theta)) \, ,
\]

\noindent since then Bayes rule reads 
\begin{equation}\label{boltzmann-bayes}
p(\theta) \, = \, \frac{1}{Z}\,  q(\theta)\, e^{\beta \, U(\theta)}.
\end{equation}

\noindent The corresponding Free Energy \eqref{freeenergy}, that is maximized by \eqref{boltzmann-bayes}, 
\begin{align} \nonumber
& \mathcal F[p(\Theta)] \, = \, \mathbb E_{p(\Theta)}[U(\Theta)] - \frac{1}{\beta} D_{KL}(p(\Theta)\|q(\Theta)) \\
& \, = \, \frac{1}{N} \mathbb E_{p(\Theta)}\Big[\log p(d|\Theta)  - \log \frac{p(\Theta)}{q(\Theta)}\Big] \, = \, -\frac{1}{N} D_{KL}\big(p(\Theta) \| q(\Theta) p(d|\Theta)\big) \label{variationalfreeenergy}
\end{align}

\noindent coincides with the \textit{variational Free Energy} $\mathcal F_{var}$ from Bayesian statistics. Indeed, from \eqref{variationalfreeenergy} it is easy to see that $\mathcal F$ assumes its maximum when $p(\Theta)$ is proportional to $q(\Theta) p(d|\Theta) $, that is when $p(\Theta)$ is given by \eqref{boltzmann-bayes}. In the literature, $\mathcal F_{var}$ is used in the variational characterization of Bayes rule, in cases when the form \eqref{boltzmann-bayes} cannot be achieved exactly but instead is approximated by optimizing \eqref{variationalfreeenergy} over the parameters $\vartheta$ of a parametrized distribution $p_\vartheta(\Theta)$ \cite{Hinton1993,Mackay1995}.
}
\end{Example}

In the following section, we will see that the Free Energy $\mathcal F$ of a bounded rational decision-making process satisfies a recursivity property, which allows the interpretation of $\mathcal F$ as a \textit{certainty-equivalent}.

\bigskip

\begin{figure}
\centering
\includegraphics[width=\textwidth]{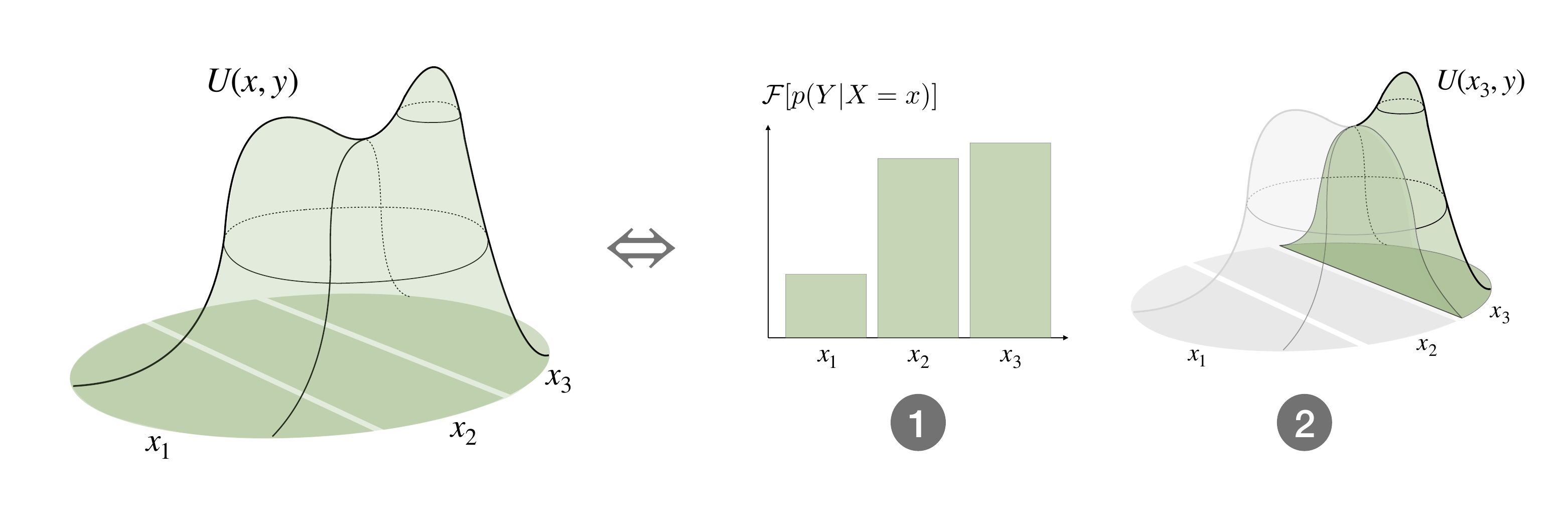}

\vspace{-10pt}
\caption{Recursivity of the Free Energy under coarse-graining. The decision about $Z=(X,Y)$ is equivalent to a two-step process consisting of the decision about $X$ and the decision about $Y$ given $X$. The objective function for the first step is the Free Energy of the second step.}
\label{fig:5}
\end{figure}

\subsection{The recursivity of $\mathcal F$ and the value of a decision problem}

Consider a bounded-optimal decision-maker with an informational cost function deciding about a random variable $Z$ with values in $\Omega$ that is decomposed into the random variables $X$ and $Y$, i.e. $Z=(X,Y)$. This decomposition can be understood as a two-step process, where first a decision about a partition $A_i$ of the full search space $\Omega=\bigcup_{i\in I} A_i$ is made, represented by a random variable $X$ with values in $I$, followed by a decision about $Y$ inside the partition selected by $X$ (see Figure \ref{fig:4a}). 

Since $p(Z) = p(X) p(Y|X)$, by the additivity of the Kullback-Leibler divergence (Proposition \ref{prop:charDKL}), we have 
\begin{align}\nonumber
& \mathcal F[p(Z)] =  \mathbb E_{p(Z)}[U(Z)] - \frac{1}{\beta} D_{KL}(p(Z)\|q(Z)) \\ \nonumber
& =  \mathbb E_{p(X)}\Big[\mathbb E_{p(Y|X)}[U(X,Y)] - \frac{1}{\beta} D_{KL}\big(p(Y|X)\|q(Y|X)\big)\Big] - \frac{1}{\beta} D_{KL}(p(X)\|q(X)) \, ,
\end{align}

\noindent and therefore, if $\mathcal F_\beta[p(Y|X)] \coloneqq \mathbb E_{p(Y|X)}[U(X,Y)] - \frac{1}{\beta} D_{KL}(p(Y|X)\|q(Y|X))$ denotes the Free Energy of the second step,
\begin{equation} \label{recursivity}
\mathcal F[p(X)p(Y|X)] \, = \, \mathbb E_{p(X)}\big[\mathcal F_\beta[p(Y|X)]\big] - \frac{1}{\beta} D_{KL}(p(X)\|q(X)) \, .
\end{equation}

\noindent In particular, the Free Energy $\mathcal F_\beta[p(Y|X)]$ of the second decision-step plays the role of the utility function of the first decision-step (see Figure \ref{fig:5}). In Equation \eqref{recursivity}, the two decision-steps have the same resource parameter $\beta$, controlling the strength of the constraint on the total informational costs 
\[
D_{KL}(p(Z)\|q(Z)) \, = \, D_{KL}(p(X)\|q(X))  + \mathbb E_{p(X)} \big[D_{KL}(p(Y|X)\|q(Y|X))\big] \, .
\] 

\noindent More generally, each step might have a separate information-processing constraint, which requires two resource parameters $\beta_1$ and $\beta_2$, and results in the total Free Energy 
\[
\mathcal F[p(X),p(Y|X)] \, = \, \mathbb E_{p(X)}\big[\mathcal F_{\beta_2}[p(Y|X)]\big] - \frac{1}{\beta_1} D_{KL}(p(X)\|q(X)) \, .
\]

\begin{Example}\label{ex:freeenergy} 
{\normalfont
Consider a bounded-rational decision-maker with informational cost function and a utility function $U$ defined on a set $\Omega = \{z_1,\dots,z_4\}$ with values as given in Figure \ref{fig:6} and an information-processing bound of $0.2$ bits ($\beta\approx 0.9$). If we partition $\Omega$ into the disjoint subsets $\{z_1,z_2\}$ and $\{z_3,z_4\}$, then the decision about $Z$ can be decomposed into two steps, $Z=(X,Y)$, the decision about $X$ corresponding to the choice of the partition and the decision about $Y$ given $X$ corresponding to the choice of $z_i$ inside the given partition determined by $X$. According to Equation \eqref{recursivity}, the choice of the partition $X=x_i$ is not in favor of the achieved expected utility inside each partition, but of the Free Energy (see Figure \ref{fig:6}). 
}
\end{Example}

Therefore, a bounded rational decision-maker that has the choice among decision-problems ideally should base its decision not on the expected utility that might be achieved but on the Free Energy of the subordinate problems. In other words, the Free Energy quantifies the value of a decision-problem, that besides of the achieved average utility also takes the information-processing costs into account.

\bigskip

\begin{figure}
\centering
\includegraphics[width=\textwidth]{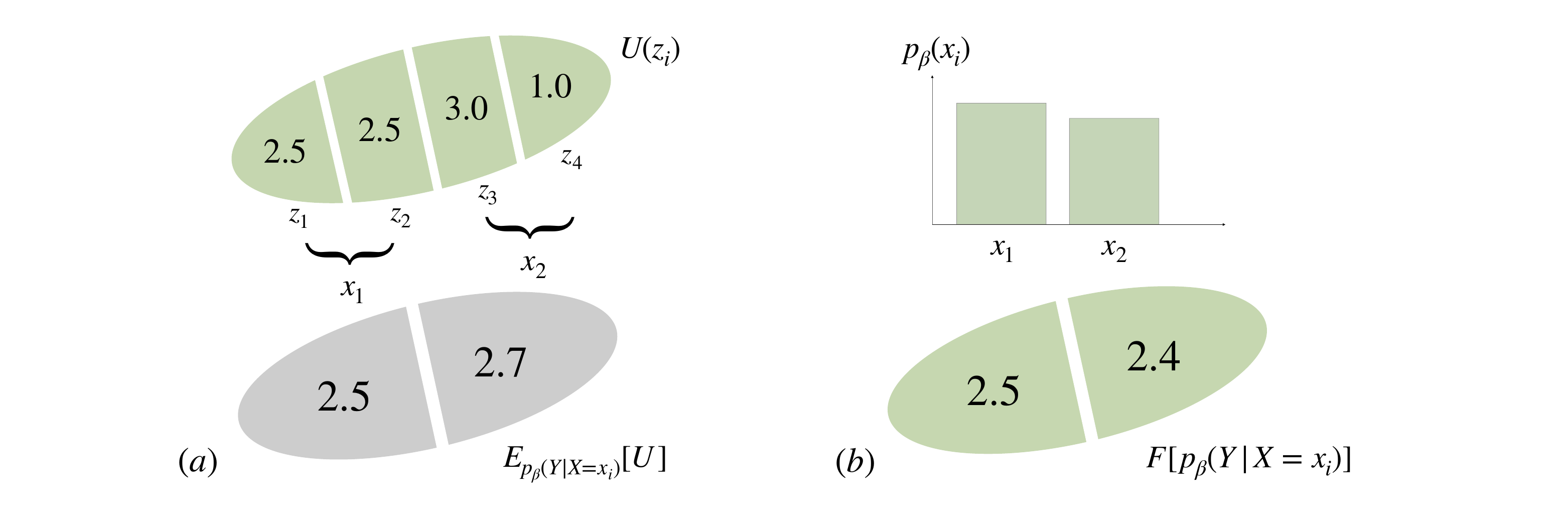}

\vspace{-5pt}
\caption{The Free Energy as certainty-equivalent (Example \ref{ex:freeenergy}). $(a)$ Utility function $U$ as a function of $z_i$ (top) and expected utilities for the coarse-grained partitions $\{z_1,z_2\}$ and $\{z_3,z_4\}$ corresponding to the choices $x_1$ and $x_2$, respectively, for a bounded-rational decision-maker with $\beta=0.9$ (bottom). $(b)$ The bounded optimal probability distribution over $x_i$ (top) does not correspond to the expected utilities in $(a)$ but to the Free Energy of the second decision-step, i.e. the decision about $Y$ given $X$ (bottom).}
\label{fig:6}
\end{figure}

\subsection{Multi-task decision-making and the optimal prior} \label{sec:multitask}

So far, we have considered decision-making problems with utility functions defined on $\Omega$ only, modelling a single decision-making task. This is extended to multi-task decision-making problems by utility functions of the form $U:\mathcal W\times \Omega\to \mathbb R, (w,x)\mapsto U(w,x)$, where the additional variable $w\in\mathcal W$ represents the current state of the world. Different world states $w$ in general lead to different optimal decisions $x^\ast(w) \coloneqq \mathrm{argmax}_{x\in \Omega} U(w,x)$. For example, in a chess game the optimal moves depend on the current board configurations the players are faced with. 

The prior $q$ for a bounded-rational multi-task decision-making problem may either depend or not depend on the world state $w\in\mathcal W$. In the first case, the multi-task decision-making problem is just given by multiple single-task problems, i.e. for each $w\in\mathcal W$, $q(X|W=w)$ and $p(X|W=w)$ are the prior and posterior of a bounded rational decision-making process with utility function $x\mapsto U(x,w)$, as described in the previous sections. In the case when there is a single prior for all world states $w\in\mathcal W$, the Free Energy is
\begin{equation} \label{freeenergy-multitask}
\mathcal F[p(X|W)] = \mathbb E_{p(W)}\Big[ \mathbb E_{p(X|W)}[U(W,X)] - \frac{1}{\beta} D_{KL}(p(X|W)\|q(X))\Big] \, ,
\end{equation}

\noindent where $p(W)$ is a given world state distribution. Note that, for simplicity we assume that $\beta$ is independent of $w\in\mathcal W$, which means that only the average information-processing is constrained, in contrast to the information-processing being constrained for each world state which in general would result in $\beta$ being a function of $w$. Similarly to single-task decision-making (Equation \eqref{boltzmann}) the maximum of \eqref{freeenergy-multitask} is achieved by
\begin{equation} \label{boltzmann-multitask}
p_\beta(x|w)  \, = \, \frac{1}{Z_\beta(w)} \, q(x) \, e^{\beta U(w,x)} \,
\end{equation}

\noindent with normalization constant $Z_\beta(w)$. Interestingly, the deliberation cost in \eqref{freeenergy-multitask} depends on how well the prior was chosen to reach all posteriors without violating the processing constraint. In fact, viewing the Free Energy \eqref{freeenergy-multitask} as a function of both, posterior \textit{and prior}, $\mathcal F[p(X|W)] = \mathcal F[p(X|W),q(X)]$, and optimizing for the prior yields the marginal of the joint distribution $p(W,X)=p(W)p(X|W)$, i.e. the mean of the posteriors for the different world states,
\begin{equation} \label{optimalprior}
q^\ast(X) \, \coloneqq \, \mathrm{argmax}_{q(X)} \mathcal F[p(X|W),q(X)] \, = \, \mathbb E_{p(W)}[p(X|W)] \, .
\end{equation}

\noindent Similarly to \eqref{boltzmann}, Equation \eqref{optimalprior} follows from finding the zeros of the functional derivative of the Free Energy with respect to $q(X)$ (modified by an additional term for the normalization constraint).

Optimizing the Free Energy $\mathcal F[p(X|W),q(X)]$ for both prior and posterior can be achieved by iterating Equations \eqref{boltzmann-multitask} and \eqref{optimalprior}. This results in an alternating optimization algorithm, originally developed independently by Blahut and Arimoto to calculate the capacity of a memoryless channel \cite{Blahut1972,Arimoto1972} (see \cite{Csiszar1984} for a convergence proof by Csisz{\'a}r and Tusn{\'a}dy). Note that 
\[
\mathcal F[p(X|W),q^\ast(X)] \, = \, \mathbb E_{p(W)p(X|W)}[U(W,X)] - \frac{1}{\beta}  \, I(W;X) \, ,
\]  

\noindent in particular, that the information-processing cost is now given by the mutual information $I(W;X)$ between the random variables $W$ and $X$. In this form, we can see that the Free Energy optimization with respect to prior and posterior is equivalent to the optimization problem in classical rate distortion theory \cite{Shannon1959}, where $U$ is given by the negative of the distortion measure. 

Similarly as in rate-distortion theory, where compression algorithms are anaylized with respect to the rate-distortion function, any decision-making system can now be analyzed with respect to informational bounded-optimality. More precisely, when plotting the achieved expected utility against the information-processing resources of a bounded-rational decision-maker with optimal prior, we obtain a pareto-optimality curve that forms an efficiency-frontier that cannot be surpassed by any decision-making process (see Figure \ref{fig:7}$(c)$).

\medskip

\subsection{Multi-task decision-making with unknown world state distribution}

\label{sec:unknownworlddistribution}
{\normalfont
A bounded rational decision-making process with informational cost and utility $U:\mathcal W\times\Omega \to \mathbb R$ that has an optimal prior $q^\ast(X)$ given by the marginal \eqref{optimalprior} must have perfect knowledge about the world state distribution $p(W)$. In contrast, here we consider the case when the exact shape of the world state distribution is unknown to the decision-maker and therefore has to be inferred from the already seen world states. More precisely, we assume that the world state distribution is parametrized by a parameter $\theta\in\mathbb R$, i.e. $p(W) = p_{\theta_{true}}(W)$ for a given parametrized distribution $p_\theta(W)$. Since the true parameter $\theta_{true}$ is unknown, $\theta$ is treated as a random variable by itself, so that $p_\theta(W) = p(W|\Theta=\theta)$. After a dataset $d = \{w_1,\dots,w_N\} \in \mathcal W^N$ of samples from $p(W|\Theta=\theta_{true})$ has been observed the joint distribution of all involved random variables can be written as
\[
p(\Theta,D,W,X) = p(\Theta) p(D|\Theta) p(W|\Theta) p(X|D,W)
\]

\noindent where $p(\Theta)$ denotes the decision-maker's prior belief about $\Theta$, and $p(D=d|\Theta) = \prod_{i=1}^N p(w_i|\Theta)$ is the likelihood of the previously observed world states. Therefore, the resulting (multi-task) Free Energy (see Equation \eqref{freeenergy-multitask}) is given by
\begin{equation}\label{freeenergy-multitaskinference}
\mathbb E_{p(\Theta)p(D|\Theta)p(W|\Theta)}\Big[\mathbb E_{p(X|D,W)}[U(W,X)] - \frac{1}{\beta} D_{KL}(p(X|D,W) \| q(X|D)) \Big] \, .
\end{equation}

It turns out that we obtain Bayesian inference as a byproduct of optimizing \eqref{freeenergy-multitaskinference} with respect to the prior $q(X|D)$. Indeed, by calculating the functional derivative with respect to $q(X|D)$ of the Free Energy \eqref{freeenergy-multitaskinference} plus an additional term for the normalization constraint of $q(X|D)$ (with Lagrange multiplier $\lambda$), we can see that any distribution $q^\ast(X|D)$ that optimizes \eqref{freeenergy-multitaskinference} must satisfy
\[
\frac{1}{\beta} \mathbb E_{p(\Theta)p(W|\Theta)}\Big[\tfrac{p(D|\Theta)p(X|D,W)}{q^\ast(X|D)}\Big] + \lambda = 0 \, ,
\]

\noindent where $\lambda\in\mathbb R$ is chosen such that $q^\ast(X|D{=}d)\in\mathbb P_{\Omega}$ for any $d\in \mathcal W^N$. This is equivalent to
\[
q^\ast(X|D) = \frac{1}{Z_D} \, \mathbb E_{p(\Theta)}\big[p(D|\Theta) \, \mathbb E_{p(W|\Theta)}[p(X|D,W)] \big] \, ,
\]

\noindent where $Z_D$ denotes the normalization constant of $q^\ast(X|D)$, given by $Z_D =  \mathbb E_{p(\Theta)}[p(D|\Theta)]$, since $\mathbb E_{p(X|D,W)}[1] = 1$ as well as $\mathbb E_{p(W|\Theta)}[1] = 1$. Therefore, we obtain
\[
q^\ast(X|D) = \mathbb E_{p^\ast(\Theta|D)}\big[\mathbb E_{p(W|\Theta)}[p(X|D,W)]\big]
\]

\noindent with $p^\ast(\Theta|D)$ as defined in Equation \eqref{bayesianinferenceoptimality}. Hence, we have shown 

\begin{Proposition}[Optimality of Bayesian inference]\label{prop:optBI}
The optimal prior $q^\ast(X|D)$ that maximizes \eqref{freeenergy-multitaskinference} is given by $q^\ast(X|D) = \mathbb E_{p^\ast(\Theta|D)p(W|\Theta)}[p(X|D,W)]$, where $p^\ast(\Theta|D)$ is the Bayes posterior
\begin{equation} \label{bayesianinferenceoptimality}
p^\ast(\Theta|D) \, \coloneqq \, \frac{p(\Theta) p(D|\Theta)}{\mathbb E_{p(\Theta)}[p(D|\Theta)] } \, .
\end{equation}

\end{Proposition}

}

\bigskip

\begin{figure}
\centering
\noindent\makebox[1.0\textwidth]{\includegraphics[width = 1.1\textwidth]{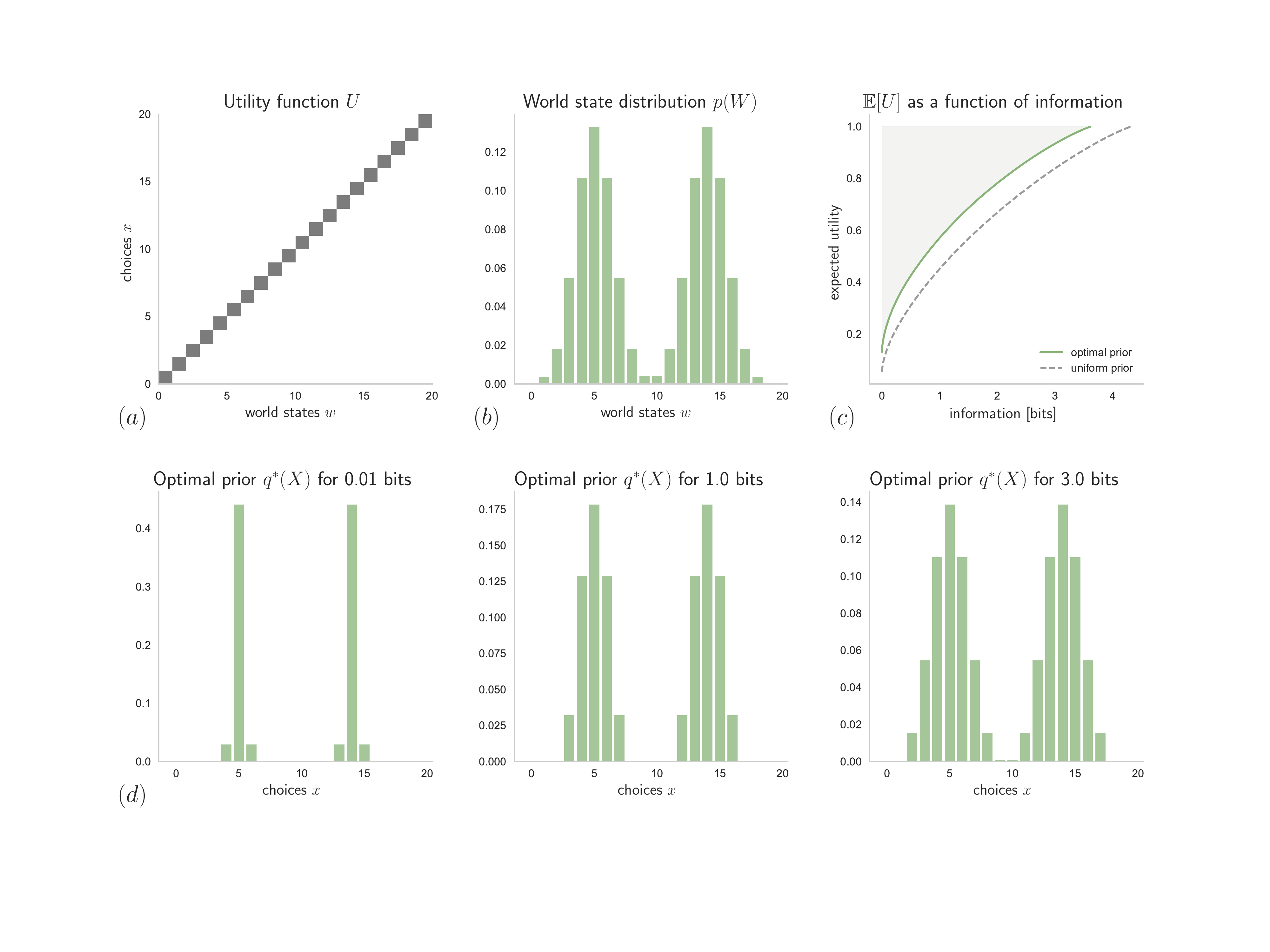}}  

\vspace{-45pt}
\caption{Absolute identification task with known world state distribution. $(a)$ Utility function, $(b)$ world states distribution (a mixture of two Gaussians), $(c)$ expected utility as a function of information-processing resources for a bounded-optimal decision-maker with a uniform and with an optimal prior (the shaded region cannot be reached by any decision-making process), and $(d)$ exemplary optimal priors $q^\ast(X)$ for different information-processing bounds.}
\label{fig:7}
\end{figure}

\section{Example: Absolute identification task with known and unknown stimulus distribution} \label{sec:example}

Consider a bounded rational decision-maker with a multi-task utility function $U$ such that for each $w\in\mathcal W$, $U(w,x)$ is non-zero for only one choice $x\in\Omega$ as shown in Figure \ref{fig:7}. Here, the decision and world spaces are both finite sets of $N=20$ elements. The world state distribution $p(W)$ is given by a mixture of two Gaussian distributions as shown in Figure \ref{fig:7}$(b)$. Due to some world states $w\in \mathcal W$ being more likely than others, there are some choices $x\in\Omega$ that are less likely to be optimal.

\subsection{Known stimulus distribution}
As can be seen in Figure \ref{fig:7}$(c)$ (dashed line), here it is not ideal to have a uniform prior distribution, $q(x) = \frac{1}{N}$ for all $x\in \Omega$. Instead, if the world state distribution is known perfectly and the prior has the form suggested by Equation \eqref{optimalprior}, i.e. $q(x) = \sum_w p(w) p_\beta(x|w)$, then, as can be seen in Figure \ref{fig:7}$(c)$ (solid line), achieving the same expected utility as with a uniform prior requires less informational resources. In particular, the explicit form of $q^\ast$ depends on the resource parameter $\beta$, see Figure \ref{fig:7}$(d)$. For low resource availability (small $\beta$), only the choices that correspond to the most probable world states are considered. However, for $\beta\to \infty$, we have 
\[
q^\ast(x) = \sum_w p(W{=}w) \delta_{w,x} = p(W{=}x) \, ,
\]

\noindent because here $\lim_{\beta\to\infty} p_\beta(x|w)= \delta_{w,x}$ is the posterior of a rational decision-maker, where $\delta_{w,x}$ denotes the Kronecker-$\delta$ (which is only non-zero if $w{=}x$). Hence, for decision-makers with abundant information-processing resources (large $\beta$) the optimal prior $q^\ast(X)$ approaches the form of the world state distribution $p(W)$ (since here $\mathcal W = \Omega$).

\smallskip

\begin{figure}
\noindent\makebox[1.0\textwidth]{\includegraphics[width = 1.1\textwidth]{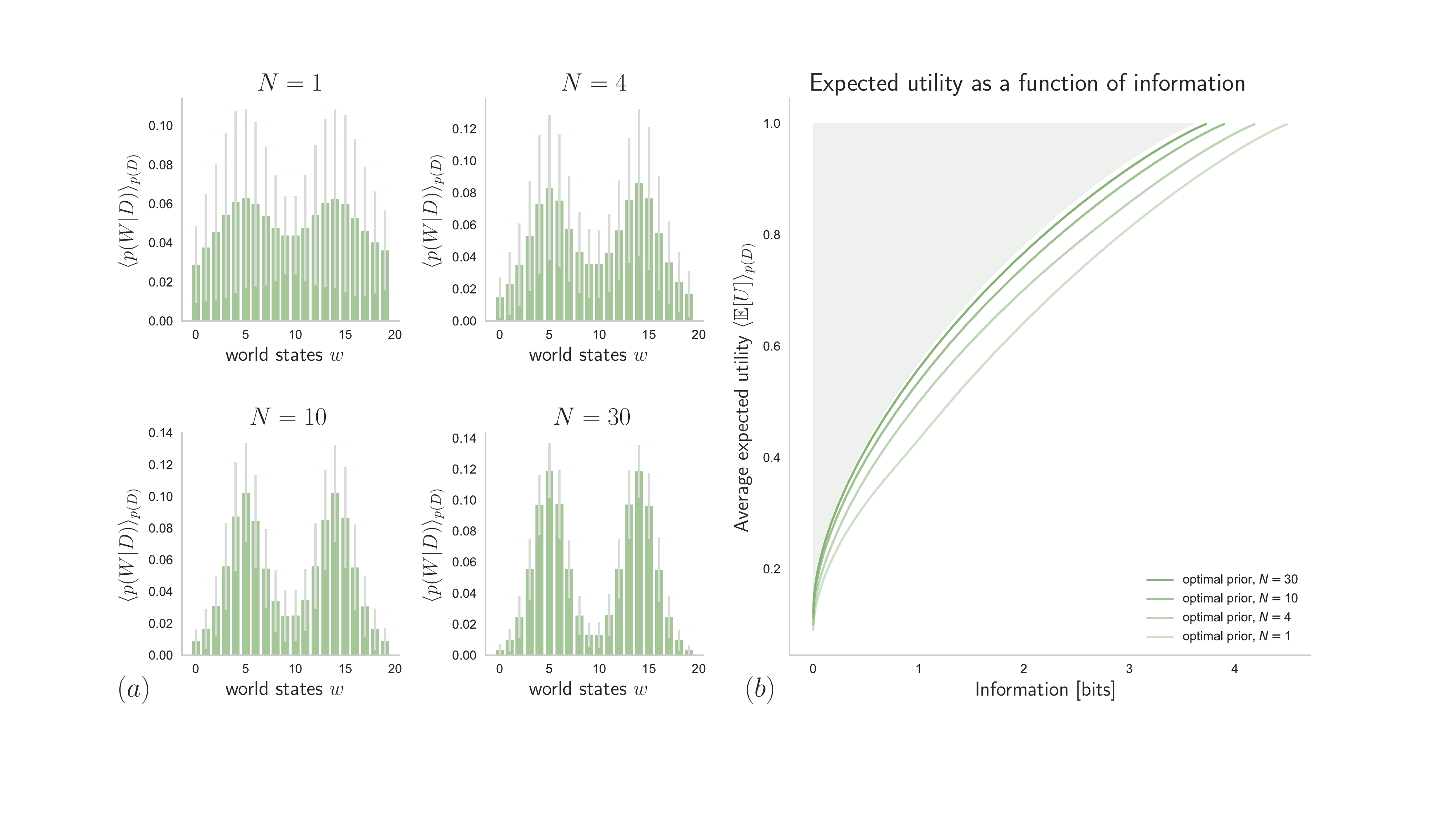}}  

\vspace{-25pt}
\caption{Absolute identification task with unknown world state distribution. $(a)$ Average of inferred world state distributions for different sizes $N$ of datasets (standard-deviations across datasets indicated by error bars). $(b)$ Resulting utility-information curves of a bounded-rational decision-maker with optimal prior that has to infer the world state distribution from datasets with different sizes $N$.}
\vspace{10pt}
\label{fig:8}
\end{figure}

\subsection{Unknown stimulus distribution}


In the case when the decision-maker has to infer its knowledge about $p(W)$ from a set of samples $d=\{w_1,\dots,w_N\}$, we know from Section \ref{sec:unknownworlddistribution} that this is optimally done via Bayesian inference. Here, we assume a mixture auf two Gaussians as a parametrization of $p(W)$, so that $\theta = (\mu_1,\mu_2,\sigma_1,\sigma_2)$, where $\mu_i$ and $\sigma_i$ denote mean and standard-deviation of the $i$-th component respectively (for simplicity, with fixed equal weights for the two mixture components).

In Figure \ref{fig:8}$(a)$, we can see how different values of $N$ affect the belief about the world state distribution, $p(W|D) = \mathbb E_{p(\Theta|D)}[p(W|\Theta)]$, when $p(\Theta|D)$ is given by the Bayes posterior \eqref{bayesianinferenceoptimality} with a uniform prior belief $p(\Theta)$. The resulting expected utilities (averaged over samples from $p(D|\theta_{true})$) as functions of available information-processing resources are displayed in Figure \ref{fig:8}$(b)$, which shows how the performance of a bounded-rational decision-maker with optimal prior and perfect knowledge about the true world state distribution is approached by bounded rational decision-makers with limited but increasing knowledge given by the sample size $N$.


Abstractly, we can view Equation \eqref{bayesianinferenceoptimality} as the bounded optimal solution to the decision-making problem that starts with a prior $p(\Theta)$ and arrives at a posterior $p(\Theta|D\,{=}\,d)$ after processing the samples in $d=\{w_1,\dots,w_N\}$ (see also Example \ref{ex:bayesianinference}). In fact, the posteriors shown in Figure \ref{fig:8}$(a)$ satisfy the requirements for a decision-making process with resource given by the number of datapoints $N$, when averaged over $p(D)$. In particular, by increasing $N$ the posteriors contain less and less uncertainty with respect to the preorder $\prec$ given by majorization. Accordingly, if we plot the achieved expected utility against the number of samples, we obtain an optimality curve similar to Figure \ref{fig:7}$(c)$ and Figure \ref{fig:8}$(b)$. In Figure \ref{fig:9} we can see how Bayesian Inference outperforms Maximum Likelihood when evaluated with respect to the average expected utility of a bounded-rational decision-maker with $2$ bits of information-processing resources.

\begin{figure}
\noindent\makebox[1.0\textwidth]{\includegraphics[width = .7\textwidth]{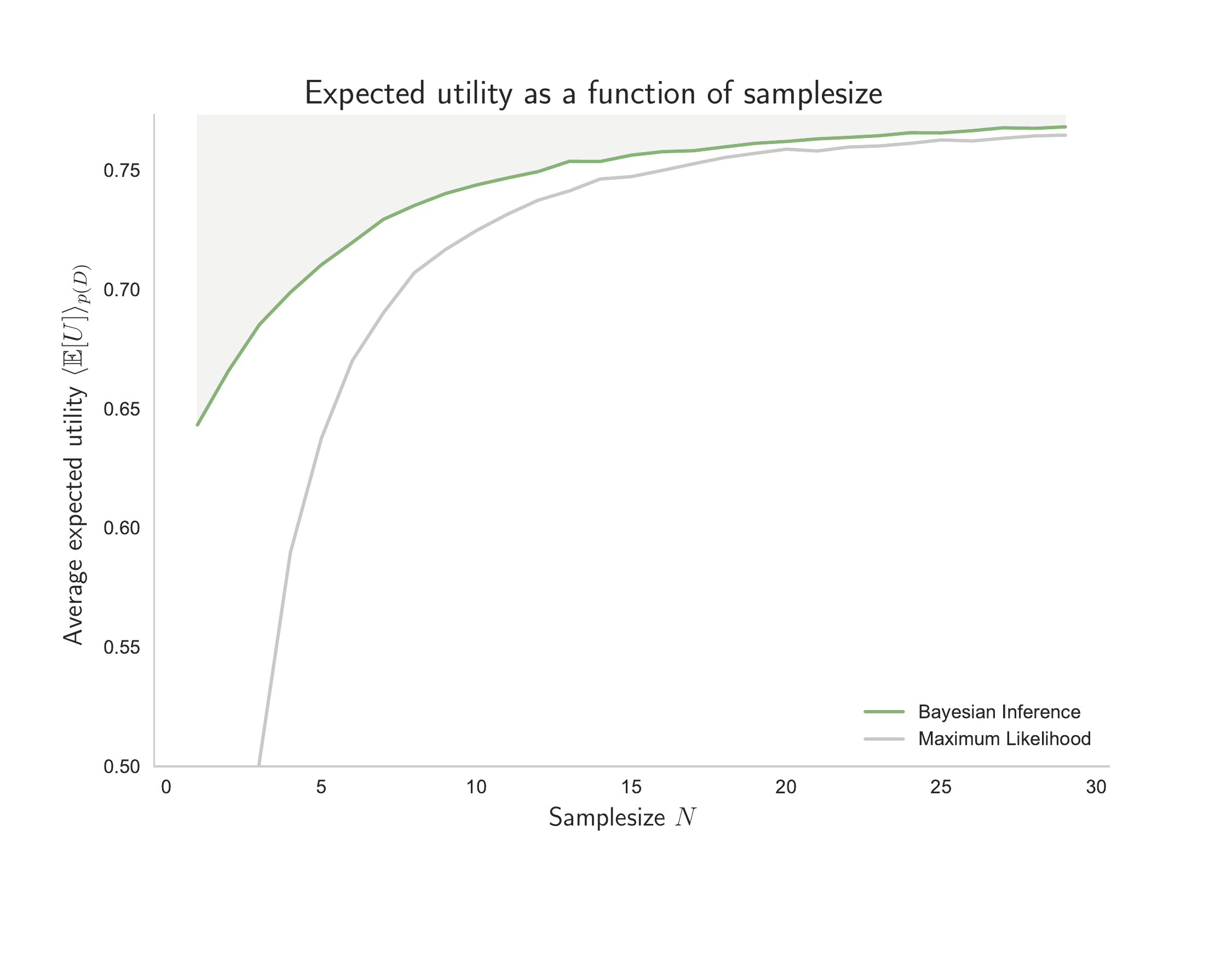}}  

\vspace{-30pt}
\caption{Optimality curve given by Bayesian inference. The average expected utility as a function of $N$ achieved by a bounded-rational decision-maker that infers the world state distribution with Bayes rule \eqref{bayesianinferenceoptimality} forms an efficiency frontier that cannot be surpassed by any other inference scheme, like for example Maximum Likelihood, when starting from the same prior belief about the world.}
\label{fig:9}
\end{figure}

\bigskip

\section{Discussion}\label{sec:discussion}

In this work, we have developed a generalized notion of decision-making in terms of uncertainty reduction. Based on the simple idea of transferring pieces of probability between the elements of a probability distribution, which we call elementary computations, we have promoted a notion of uncertainty which is known in the literature as majorization, a preorder $\prec$ on $\mathbb P_\Omega$. Taking non-uniform initial distributions into account, we extended the concept to the notion of relative uncertainty, which corresponds to relative majorization $\prec_q$. Despite of the large amount of research that has been done on majorization theory, from the early works \cite{Muirhead1902,Lorenz1905,HLP1934} through further developments \cite{Ruch1976,Arnold1987,Ando1989,Pecaric1992,Cohen1993,Bhatia1997,Marshall2011} to modern applications \cite{Brandao2013,Horodecki2013,Gour2015}, there is a lack of results on the more general concept of relative majorization. This does not seem to be due to a lack of interest, as can be seen from the results \cite{Veinott1971,Ruch1978,Joe1990,Latif2017}, but mostly because relative majorization looses some of the appealing properties of majorization which makes it harder to deal with, for example that permutations no longer leave the ordering $\prec_q$ invariant, in contrast to the case of a uniform prior. This restriction does, however, not affect our application of the concept to decision-making, as permutations are not considered as elementary computations, since they do not diminish uncertainty. By reducing the non-uniform to the uniform case, we managed to prove new results on relative majorization (Theorem \ref{thm:charrel}), which then enabled new results in other parts of the paper (Example \ref{ex:genDiv}, Propositions \ref{prop:superadditivityDiv}, \ref{prop:boundedoptDM}), and allowed an intuitive interpretation of our final definition of a decision-making process (Definition \ref{def:DM}) in terms of elementary computations with respect to non-uniform priors (Definition \ref{def:relpigoudalton}).

More precisely, starting from stepwise elimination of uncertain options (Section \ref{sec:probDMprior}), we have argued that decision-making can be formalized by transitions between probability distributions (Section \ref{sec:probDM}), and arrived at the concept of decision-making processes  traversing probability space from prior to posterior by successively moving pieces of probability between options such that uncertainty relative to the prior is reduced (Section \ref{sec:elim}). Such transformations can be quantified by cost functions, which we define as order-preserving functions with respect to $\prec_q$ and capture the resource costs that must be expended by the process. We have shown (Propositions \ref{prop:superadditivity}, \ref{prop:superadditivityDiv}) that many known generalized entropies and divergences, which are examples of such cost functions (Examples \ref{ex:genEnt} and \ref{ex:genDiv}), satisfy superadditivity with respect to coarse-graining. This means that under such cost functions, decision-making costs can potentially be reduced by a more intelligent search strategy, in contrast to Kullback-Leibler divergence, which was characterized as the only additive cost function (Proposition \ref{prop:charDKL}). There are plenty of open questions for further investigation in that regard. First, it is not clear under which assumptions on the cost functions $C_q$ superadditivity could be improved to $C_q(p) = \alpha D_{KL}(p\|q) + r(p,q)$ with $\alpha>0$ and $r(p,q)\geq 0$. Additionally, it would be an interesting challenge to find  sufficient conditions implying super-additivity that include more cost functions than $f$-divergences. The field of information geometry might give further insights on the topic, since there are studies in similar directions, in particular characterizations of divergence measures in terms of information monotonicity and the data-processing inequality \cite{Jiao2014,Amari2009,Amari2010,Amari2018}. One interesting result is the characterization of Kullback-Leibler divergence as the single divergence measure being both an $f$-divergence and a Bregman divergence.

In Section \ref{sec:bounded}, bounded rational decision-makers were defined as decision-making processes that are maximizing utility under constraints on the cost function, or equivalently minimizing resource costs under a minimal utility requirement. In the important case of additive cost functions (i.e.~proportional to Kullback-Leibler divergence), this leads to information-theoretic bounded rationality \cite{Ortega2010,Tishby2011,Ortega2013,Genewein2015,Gottwald2018}, which has precursors in the economic and game-theoretic literature \cite{McKelvey1995,Ochs1995,Mattsson2002,Wolpert2006,Spiegler2011,Howes2009,Todorov2009,Still2009,Tishby2011,Kappen2012,Vul2014,Lewis2014}. We have shown that the posteriors of a bounded rational decision-maker with increasing informational constraints leave a path in probability space that can itself be considered an anytime decision-making process, in each step perfectly trading off utility against processing costs (Prop.~\ref{brdmproc}). In particular, this means that the path of a bounded rational decision-maker with informational cost decreases uncertainty with respect to \textit{all} cost functions, not just Kullback-Leibler divergence. We have also studied the role of the prior in bounded rational multi-task decision-making, where we have seen that imperfect knowledge about the world state distribution leads to Bayesian inference as a byproduct, which is in line with the characterization of Bayesian inference as minimizing prediction surprise \cite{Ortega2010b}, but also demonstrates the wide applicability of the developed theory of decision-making with limited resources.

Finally, in Section \ref{sec:example}, we have presented the results of a simulated bounded rational decision-maker solving an absolute identification task with and without knowledge about the world state distribution. Additionally, we have seen that Bayesian inference can be considered a decision-making process with limited resources by itself, where the resource is given by the number of available datapoints.

\section{Conclusion}\label{sec:conclusion}

To our knowledge, this is the first principled approach to decision-making based on the intuitive idea of Pigou-Dalton-type probability transfers (elementary computations). Information-theoretic bounded rationality has been introduced by other axiomatic approaches before \cite{Ortega2010, Mattsson2002}. For example, in \cite{Ortega2010} a precise relation between rewards and information value was derived by postulating that systems will choose those states with high probability that are desirable for them. This leads to a direct coupling of probabilities and utility, where utility and information inherit the same structure, and only differ with respect to normalization---see \cite{Candeal2001} for similar ideas. In contrast, we assume utility and probability to be independent objects a priori that only have a strict relationship in the case of bounded-optimal posteriors. The approach in \cite{Mattsson2002} introduces Kullback-Leibler divergence as disutility for decision control. Based on Hobson's characterization \cite{Hobson1969}, the authors argue that cost functions should be monotonic with respect to uniform distributions (property \eqref{uniformmonotonicity_main}) and invariant under decomposition, which coincides with additivity under coarse-graining (see Examples \ref{ex:genEnt}, \ref{ex:genDiv}). Both of these assumptions are special cases of our more general treatment, where cost functions must be monotonic with respect to elementary computations and are generally not restricted to being additive.

In the literature, there are many mechanistic models of decision-making that instantiate decision-making processes with limited resources. Examples include reinforcement learning algorithms with variable depth \cite{Keramati2011,Keramati2016}, Markov chain Monte Carlo (MCMC) models where only a certain number of samples can be evaluated \cite{Vul2014, Ortega2014, Hihn2018}, and evidence accumulation models  that accumulate noisy evidence until either a fixed threshold is reached \cite{Laming1968, Ratcliff1978, Townsend1983, Ratcliff2013, Shadlen2006} or where thresholds are determined dynamically by explicit cost functions depending on the number of allowed evidence accumulation steps \cite{Frazier2007, Drugowitsch2012}. 
Many of these concrete models may be described abstractly by resource parameterizations (Def.~\ref{def:relcostfunctions}). More precisely, in such cases the posteriors $\{p_r\}_{r\in I}\subset\Gamma\subset \mathbb P_\Omega$ are generated by an explicit process with process constraints $\Gamma$ and resource parameter $r$. For example, in diffusion processes $r$ may correspond to the amount of time allowed for evidence accumulation, in Monte Carlo algorithms $r$ may reflect the number of MCMC steps, and in a reinforcement learning agent $r$ may represent the number of forward-simulations. If the resource restriction is described by a monotonic cost function $r\mapsto c_r$ \cite{Frazier2007,Drugowitsch2012}, then the process can be optimized by a maximization problem of the form 
\begin{align} \nonumber
\max_{r\in I, p\in\Gamma_r} \big\{\mathbb E_{p}[U] - c_{r}\big\} \ & = \ \max_{r\in I, p\in\Gamma_r} \big\{ E_{p}[U] \, \big| \, c_{r}\leq M \big\} \ = \ \max_{p\in\Gamma} \big\{ E_{p}[U] \, \big| \, C_{q}(p) \leq M' \big\}
\end{align}

\noindent where $M,M'$ are non-negative constants, $\Gamma_r\subset\Gamma$ denotes the subset of probability distributions with resource $r$, and $C_q$ denotes a cost function such that $r\mapsto C_q(p)$ for $p\in\Gamma_r$ is strictly monotonically increasing. In particular, such cases can also be regarded as  bounded rational decision-making problems of the form \eqref{eq:optU}.

Bounded rationality models in the literature come in a variety of flavors. 
In the heuristics and biases paradigm the notion of optimization is often dismissed in its entirety \cite{Gigerenzer2001}, even though decision-makers still have to have a notion of options being better or worse, for example in order to adapt their aspiration levels in a satisficing scheme \cite{Simon1982}. We have argued that from an abstract normative perspective we can formulate satisficing in probabilistic terms, such that one could investigate the efficiency of heuristics within this framework.
Another prominent approach to bounded rationality is given by systems capable of decision-making about decision-making, i.e. meta-decision-making. Explicit decision-making processes composed of two decision steps have been studied, for example, in the reinforcement learning literature \cite{Keramati2011, Keramati2016, Pezzulo2013, Viejo2015}, where the first step is represented by a meta decision about whether a cheap model-free or a more expensive model-based learning algorithm is used in the second step. The meta step consists of a trade-off between the estimated utility against the decision-making costs of the second decision step. 
In the information-theoretic framework of bounded rationality this could be seen as a natural property of multi-step decision-making and the recursivity property \eqref{recursivity} from which follows that the value of a decision-making problem is given by its Free Energy, that besides of the achieved utility also takes the corresponding processing costs into account. Another prominent approach to bounded rationality is computational rationality \cite{Lewis2014} where the focus lies on finding bounded-optimal programs that solve constrained optimization problems presented by the decision-maker's architecture and the task environment. As described above, such architectural constraints could be represented by a process dependent subset $\Gamma\subset \mathbb{P}_\Omega$, and in fact our resource costs could be included into such a subset $\Gamma_r$ as well. From this point of view, both frameworks would look for bounded-optimal solutions in that the search space is first restricted and then the best solution in the restricted search space is found. However, our search space would consist of distributions describing probabilistic input-output maps, whereas the search space of programs would be far more detailed.

The notion of decision-making presented in this work, intuitively developed from the basic concept of uncertainty reduction given by elementary computations and motivated by the simple idea of progressively eliminating options, on the one hand provides a promising theoretical playground that is open for further investigation (e.g.~superadditivity of cost functions and minimality of relative entropy), potentially providing new insights into the connection between the fields of rationality theory and information theory, and on the other hand serves the purpose of a general framework to describe and analyze all kinds of decision-making processes (e.g.~in terms of bounded-optimality).


\bigskip 

\begin{appendix}

\section{Proofs of technical results from Sections \ref{sec:1} and \ref{sec:bounded}}

\medskip

\begin{Proposition}[Superadditivity of generalized entropies under coarse-graining, Example \ref{ex:genEnt}] \label{prop:superadditivity}
All cost functions of the form 
\begin{equation}\label{genEnt_app}
C(p) \, = \, \sum\nolimits_{i=1}^N f(p_i) \, ,
\end{equation}

\noindent with $f$ (strictly) convex and differentiable on $[0,1]$, and $f(1) = 0$, are superadditive with respect to coarse-graining, that is 
\[
C(Z) \geq C(X) + C(Y|X)
\]

\noindent whenever $Z=(X,Y)$, and $C(X)\coloneqq C(p(X))$ and $C(Y|X)\coloneqq \mathbb E_{p(X)}[C(p(Y|X))]$.

\end{Proposition} 

\begin{proof}
As we will see in the following proof, strict convexity is not needed for superadditivity, but it is required for the definition of a cost function. First, since $\sum_i p_i = 1$, notice that we can always redefine the convex function $f$ in \eqref{genEnt_app} by $f_c(t)\coloneqq f(t)-c(t-1)$ for an arbitrary constant $c\in\mathbb R$ without changing $C(p)$ for all $p\in\mathbb P_\Omega$. Hence, without loss of generality, we may assume $f'(1) = 0$, so that $t=1$ is a global minimum of $f$ (since $f(t)\geq f(1) + (t-1)\,f'(1) = f(1)$ due to convexity). Since $C$ is symmetric, superadditivity under coarse-graining is equivalent to 
\begin{equation} \label{recursivity_app}
C(p_1,\dots,p_N) \geq C(p_1+p_2,p_3,\dots,p_N) + (p_1{+}p_2)\, C(\tfrac{p_1}{p_1+p_2},\tfrac{p_2}{p_1+p_2}) \end{equation}

\noindent This simply follows by induction, since \eqref{recursivity_app} corresponds to the partitioning $\Omega =\bigcup_{j=1}^{N-1} A_j$ with $A_1 = \{x_1,x_2\}$ and $A_j = \{x_{j+1}\}$ for all $j=2,\dots,N-1$ (see also \cite[2.3.5]{Aczel1975}). In terms of $f$, \eqref{recursivity_app} reads
\[
f(p_1)+f(p_2) \, \geq \, f(p_1+p_2) + (p_1+p_2) \, \Big(f\big(\tfrac{p_1}{p_1+p_2}\big) + f\big(\tfrac{p_2}{p_1+p_2} \big)  \Big)
\]

\noindent By setting $u = p_1+p_2$ and $v = \tfrac{p_1}{p_1+p_2}$, this is equivalent to
\begin{equation} \label{rec_eq}
f(uv) + f(u(1-v)) \, \geq \, f(u) + u \, \big(f(v)+f(1-v) \big) 
\end{equation}

\noindent for all $u,v\in[0,1]$. Writing $F_v(u)\coloneqq f(uv) + f(u(1-v))-f(u) - u \, (f(v)+f(1-v))$ and noting that $F_v(0)\geq 0$ and $F_v(1) = 0$, it suffices to show that $F'_v(u) \leq 0$, which shows that $F_v$ is monotonically decreasing from $F_v(0)$ to $F_v(1)=0$ and thus $F_v(u)\geq 0$ for all $u\in[0,1]$. We have for all $v,u\in[0,1]$
\[
F'_v(u) = v f'(uv) + (1{-}v) f'(u(1{-}v)) - f'(u) - f(v) - f(1{-}v)
\]

\noindent By the symmetry of $F_v$ under the replacement of $v$ by $1{-}v$, without loss of generality, we may assume that $v\leq \frac{1}{2}$, so that $u v \leq u(1{-}v) \leq u$. Since $f$ is convex, $f'$ is monotonically increasing on $[0,1]$, and thus $f'(uv)\leq f'(u(1{-}v) \leq f'(u)$. In particular,
\[
f'(u)  = v f'(u) + (1-v) f'(u) \geq v f'(uv) + (1-v) f'(u(1{-}v))
\]

\noindent and thus, since $f(v)+f(1{-}v) \geq 0$, it follows that $F'_v(u) \leq 0$, which completes the proof.
\end{proof}

\bigskip

\begin{Proposition}[Characterization of Shannon entropy, Example \ref{ex:genEnt}] \label{prop:charshannon}
If a cost function $C$ is additive under coarse-graining, that is if $C(Z) = C(X) + C(Y|X)$ with the notation from Proposition \ref{prop:superadditivity}, then 
\[
C = -\alpha H
\] 

\noindent for some $\alpha\geq 0$, i.e. $C$ is proportional to the negative Shannon entropy $-H$.
\end{Proposition}

\vspace{-15pt}
\begin{proof}
Since uniform distributions over $N$ options are majorized by uniform distributions over $N'<N$ options (see \eqref{precuniform}), it follows for any cost function $C$ that the function $f$ defined by $f(N) := C\big(\tfrac{1}{N},\dots,\tfrac{1}{N}\big)$ is monotonically increasing. Therefore the claim follows directly from Shannon's proof \cite{Shannon1948}, where he shows that this monotonicity, additivity under coarse-graining, and continuity determine Shannon entropy up to a constant factor. 
\end{proof}

\medskip

\begin{Proposition}[Prop.\ref{prop:unifminimal_main}] \label{prop:unifminimal} The uniform distribution $(\tfrac{1}{N},\dots,\tfrac{1}{N})$ is the unique minimal element in $\mathbb P_\Omega$ with respect to $\prec$, i.e.

\vspace{-5pt}
\begin{equation}\label{unifminimal}
\big(\tfrac{1}{N},\dots,\tfrac{1}{N} \big) \, \prec \, p \qquad \forall p\in\mathbb P_\Omega \, .
\end{equation}

\end{Proposition}

\begin{proof}
For the proof of \eqref{unifminimal}, let $(\Pi_i)_{i=1}^N$ be the family of all cyclic permutations of the $N$ entries of a probability vector $p$, such that 
\[
\Pi_1(p)=p,\quad \Pi_2(p) = (p_N,p_1,\dots, p_{N-1}),\quad \dots \quad, \quad \Pi_N(p) = (p_2,\dots,p_N,p_1)\, .
\]

\noindent It follows that $\sum_{i=1}^N\Pi_i(p) = e$ for all $p\in\mathbb P_\Omega$, where $e=(1,\dots,1)$ as above, and therefore the uniform distribution $\frac{1}{N} e$ is given by a convex combination of permutations of $p$, so that $(iv)$ in Theorem \ref{thm:characterization} implies $\frac{1}{N} e \prec p$. There is many different ways to prove uniqueness. A direct way is to assume there exists $q\in\mathbb P$ with $q\prec \frac{1}{N}e$ for all $p\in\mathbb P_\Omega$, so that by $(iii)$ in Theorem \ref{thm:characterization} there exists a stochastic matrix $A$ with $\frac{1}{N}eA = q$. But since $eA = e$ ($A$ stochastic), it follows that $q = \frac{1}{N}e$. An indirect way would be to use that if $q\prec p$ for all $p\in\mathbb P_\Omega$, then 
from Example \ref{ex:genEnt} we know that this implies $H(q)\geq H(p)$ for all $p\in\mathbb P_\Omega$, where $H$ denotes the Shannon entropy, simply because $-H$ is a cost function. In particular, $q$ maximizes $H$ and therefore coincides with the uniform distribution $\frac{1}{N}\,e$.
\end{proof}

\bigskip

\begin{Proposition}[Equivalence of $(i),(iii),(v)$ in Theorem \ref{thm:charrel}]\label{prop:charrel} The following are equivalent

\begin{enumerate}

\item[$(i)$] $p'\prec_q p$, i.e. $p'$ contains more uncertainty relative to $q$ than $p$ (Def. \ref{def:reluncertainty})\\[-4pt]

\item[$(iii)$] $p' = pA$ for a $q$-stochastic matrix $A$, i.e.~$Ae=e$ and $qA = q$\\[-4pt]

\item[$(v)$] $\sum_{i=1}^{l-1} (p'_i)^{\downarrow} + a_q(k,l) (p'_l)^{\downarrow} \ \leq \ \sum_{i=1}^{l-1}p_i^\downarrow + a_q(k,l) p_l^\downarrow$ for all $\alpha \sum_{i=1}^{l-1} q_i^\downarrow \leq k \leq \alpha \sum_{i=1}^l q_i^\downarrow$ and $1\leq l \leq N$, where $a_q(k,l)\coloneqq (\frac{k}{\alpha} - \sum_{i=1}^{l-1} q_i^\downarrow)/q_l^\downarrow$, and the arrows indicate that $(p_i^\downarrow/q_i^\downarrow)_i$ is ordered decreasingly.

\end{enumerate}
\end{Proposition}

\begin{proof} We will make use of the fact that $\Lambda_q:\mathbb P_\Omega\to\mathbb P_{\tilde \Omega}$ has a left inverse $\Lambda_q^{-1}: \Lambda_q(\mathbb P_\Omega) \to \mathbb P_{\Omega}$ satisfying $\Lambda_q^{-1}\Lambda_q = \mathbb I$, where $\mathbb I$ denotes the identity on $\mathbb P_\Omega$. This can be verrified by simply multiplying the corresponding matrices, given by

\[
\Lambda_q^{-1} = \underbrace{\begin{pmatrix} 
	1 	& \cdots & 1 	& 0 	& \cdots	& \cdots & \cdots & \cdots & 0 	\\ 
	0 	& \cdots & 0 	& 1 	& \cdots & 1 	& 0 	& \cdots & 0 	\\ 
	\ 	& \ & \ & \	& \ & \ddots	& \		& \ & \	\\ 
	\ 	& \ & \ & \	& \ & \	& \ddots		& \ & \	\\ 
	0 	& \cdots & \cdots & \cdots	& \cdots & 0 	& 1		& \cdots & 1 	\\ 
	\end{pmatrix}}_{|A_1|+\dots+|A_N| = \alpha} \]
and
\[
\Lambda_q = \frac{1}{\alpha} \left.\begin{pmatrix} 
	\tfrac{1}{q_1} 	& 0 			& \cdots 			& 0 		& 0 \\ 
	\vdots 			& \vdots 		& \ 				& \vdots		 		& \vdots \\ 
	\tfrac{1}{q_1} 	& 0 			& \cdots 			& 0		& 0 \\ 
	\ 				& \ 			& \ddots 		& \ 			& 	\ \\
	0 				& 0 			&  \cdots 		& 0 			& \frac{1}{q_N} \\ 
	\vdots 			& \vdots 		& \		 		& \vdots 		& \vdots		\\ 
	0 				& 0 			&  \cdots 		& 0 			& \frac{1}{q_N}
	\end{pmatrix}\right\}{\alpha} \, , 
\]

\vspace{5pt}
\noindent and noting that, by definition, $\alpha \, q_i = |A_i|$. 

Characterization $(v)$ follows from $(vi)$ of Theorem \ref{thm:characterization} and Definition \ref{def:reluncertainty}, since $p'\prec_q p$ if and only if
\[
\sum_{i=1}^k (\Lambda_q p')_i^\downarrow  \leq \sum_{i=1}^k (\Lambda_q p)_i^\downarrow
\] 

\noindent for all $1\leq k\leq {\alpha-1}$, from which $(v)$ is an immediate consequence. 

$(i) \Rightarrow (iii)$: Assuming that $p'\prec_q p$ we have $\Lambda_q p' \prec \Lambda_q p$ and therefore, by $(iii)$ in Theorem \ref{thm:characterization}, there exists a doubly stochastic matrix $B$ such that $\Lambda_q p' = B^T \Lambda_q p$. With $A^T\coloneqq \Lambda_q^{-1} B^T \Lambda_q$ it follows that 
\[
A e =  \Lambda_q^T B (\Lambda_q^{-1})^T e = (\Lambda_q^{-1})^T B e = (\Lambda_q^{-1})^T e = e \, ,
\]

\noindent where we have used that $(\Lambda_q^{-1})^T e = e$ and $\Lambda_q^T e=e$ which is easy to check, and $Be=e$ from $B$ being a stochastic matrix. Note that, by slightly abusing notation, here $e$ is always the constant vector $(1,\dots,1)$ irrespective of the number of its entries ($N$ or $\alpha$, depending on which space the operator is defined). Moreover, we have
\[
A^T q  = \Lambda_q^{-1} B^T \Lambda_q q =  \tfrac{1}{\alpha} \, \Lambda_q^{-1} B^T e = \tfrac{1}{\alpha} \, \Lambda_q^{-1} e = q 
\]

\noindent where we have used that $\Lambda_q q$ by definition is the uniform distribution on $\mathbb P_{\tilde \Omega}$, i.e. $\Lambda_q q = \frac{1}{\alpha}\, e$, and therefore also $\Lambda_q^{-1} e = \alpha \, q$. In particular, since also $A_{ij}\geq 0$ ($B,\Lambda_q,\Lambda_q^{-1}$ have non-negative entries), it follows that $A$ is a $q$-stochastic matrix such that $p' = A^Tp = pA$.

$(i) \Leftarrow (iii)$: Similarly, if $A$ is a $q$-stochastic matrix with $p' = pA$, then $\Lambda_q p'  = B^T \Lambda_q p$, where $B\coloneqq (\Lambda_q^{-1})^T A \Lambda_q^T$ satisfies $B^Te = \alpha \, \Lambda_q A^T q = \alpha \, \Lambda_q q = e$ and $Be = (\Lambda_q^{-1})^T A e = (\Lambda_q^{-1})^T e = e$, where we have used that $\Lambda_q^{-1} e = \alpha q$, $\Lambda_q q = \frac{1}{\alpha}\, e$, $\Lambda_q^T e = e$, and $(\Lambda_q^{-1})^T e = e$. In particular, since also $B_{ij}\geq 0$, $B$ is doubly stochastic and therefore $\Lambda_q p' \prec \Lambda_q p$, i.e.~$p'\prec_q p$. 
\end{proof}

\medskip

\begin{Proposition}[Prop.~\ref{priorminimal_main}] \label{priorminimal}
The prior $q\in\mathbb P_\Omega$ is the unique minimal element in $\mathbb P_\Omega$ with respect to $\prec_q$, that is $q\prec_q p$ for all $p\in\mathbb P_\Omega$.
\end{Proposition}
\begin{proof} Let $p\in\mathbb P_\Omega$, and let $P:=\Lambda_q p$ denote its representation in $\mathbb P_{\tilde \Omega}$. Then $Q\prec P$ by Proposition \ref{prop:unifminimal} (uniform distributions are minimal) and therefore $q = \Lambda_q^{-1} Q \prec_q p$, in particular $q$ is minimal with respect to $\prec_q$. For uniqueness, let $p'$ be possibly another minimal element. Then $p'\prec_q q$ and therefore by $(iii)$ in Theorem \ref{thm:charrel} there exists a $q$-stochastic matrix $A$ with $p' = qA$. But since $A$ is $q$-stochastic, $qA=q$, and thus $p' = q$. 
\end{proof}

\medskip

\begin{Proposition}[Example \ref{ex:genDiv}: Superadditivity of $f$-divergences under coarse-graining] \label{prop:superadditivityDiv}
All relative cost functions of the form 
\begin{equation}\label{genDiv_app}
C_q(p) = \sum_{i=1}^N q_i \, f\Big(\frac{p_i}{q_i} \Big) ,
\end{equation}

\noindent with $f$ (strictly) convex and differentiable on $[0,1]$, and $f(1) = 0$, are superadditive with respect to coarse-graining, that is, for $Z=(X,Y)$,
\[
C_q(Z) \geq C_q(X) + C_q(Y|X)
\]

\noindent whenever $C_q(X)\coloneqq C_{q(X)}(p(X))$ and $C_q(Y|X)\coloneqq \mathbb E_{p(X)}[C_{q(Y|X)}(p(Y|X))]$, which for cost functions that are symmetric with respect to permutations of $((p_i,q_i))_{i=1,\dots,N}$ (such as \eqref{genDiv}) is equivalent to 
\begin{equation}\label{relsuperadditivitysymmetric}
C_{q}(p) \ \geq \ C_{(q_1+q_2,q_3\dots,q_N)}(p_1{+}p_2,p_3,\dots,p_N) + (p_1{+}p_2) \,C_{(\frac{q_1}{q_1+q_2},\frac{q_2}{q_1+q_2})}\big(\tfrac{p_1}{p_1+p_2},\tfrac{p_2}{p_1+p_2}\big) \, .
\end{equation}

\end{Proposition} 

\begin{proof}

This is a simple corollary to Proposition \ref{prop:superadditivity}, after establishing the following interesting property of cost functions of the form \eqref{genDiv_app}:
\begin{equation} \label{Lambdainvariance}
C_q(p) = C_{\Lambda_q q}(\Lambda_q p)
\end{equation}

\noindent where $\Lambda_q$ denotes the representation mapping defined in Section \ref{sec:probDMprior} that maps $q$ to a uniform distribution on an elementary decision space $\tilde \Omega$ of $|\tilde \Omega| = \alpha$ elements, given by $(\Lambda_q p)(\omega) = \frac{1}{\alpha} \frac{p_i}{q_i} $ whenever $\omega\in A_i$, where $\{A_i\}_{i=1}^N$ is a disjoint partition of $\tilde \Omega$ such that $|A_i| = \alpha\, q_i$. Equation \eqref{Lambdainvariance} then follows from
\[
C_{\Lambda_q q}(\Lambda_q p) = \sum_{\omega\in\tilde \Omega} (\Lambda_q q)(\omega) \, f \left( \frac{(\Lambda_q p)(\omega)}{(\Lambda_q q)(\omega)} \right) \, =\, \frac{1}{\alpha} \sum_{i=1}^{N} \sum_{w\in A_i} f\left(\frac{p_i}{q_i} \right) = \sum_{i=1}^N q_i f\left(\frac{p_i}{q_i} \right) \, ,
\]

\noindent where we have used that $\sum_{w\in A_i} 1 = |A_i| = \alpha q_i$. Hence the case of a non-uniform prior reduces to the case of a uniform prior, which was shown in Proposition \ref{prop:superadditivity}.
\end{proof}

\medskip

\begin{Proposition}[Example \ref{ex:genDiv}: Characterization of Kullback-Leibler divergence]\label{prop:charDKL}
If $C_q$ is a continuous cost function relative to $q$ that is additive under coarse-graining, that is $C_q(Z) = C_q(X) + C_q(Y|X)$ in the notation of Proposition \ref{prop:superadditivityDiv}, then 
\begin{equation} \label{DKLuniqueadditivecost}
C_q(p) =  \alpha \, D_{KL}(p\|q)
\end{equation}

\noindent for some $\alpha \geq 0$, where $D_{KL}(p\|q)$ denotes the Kullback-Leibler divergence $D_{KL}(p\|q)= \sum_{i}p_i \log(p_i/q_i)$.
\end{Proposition}

\begin{proof} First, we show that any relative cost function that is additive under coarse-graining satisfies \eqref{uniformmonotonicity_main}, the monotonicity property for uniform distributions: If $f(M,N)$ denotes the cost $C_{u_N}(u_M)$ of a uniform distribution $u_M$ over $M$ elements relative to a uniform distribution $u_N$ over $N\geq M$ elements, then
\begin{equation} \label{uniformmonotonicity}
\begin{array}{rl}
f(M',N) \leq f(M,N) &\ \forall M \leq M' \leq N \, ,\\[5pt]
f(M,N) \geq f(M,N') &\  \forall M \leq N' \leq N \, . 
\end{array}
\end{equation}

\noindent Once \eqref{uniformmonotonicity} has been established, then \eqref{DKLuniqueadditivecost} goes back to a result by Hobson \cite{Hobson1969} (see also \cite{Mattsson2002}), whose proof is a modification of Shannon's axiomatic characterization \cite{Shannon1948}. 

The first property in \eqref{uniformmonotonicity} actually is true for all relative cost functions: For $q=u_N$ with $N=|\Omega|$, we have $p'\prec_q p$ iff $p' \prec p$ and thus the first property follows from \eqref{precuniform}, and the same is true in the case when $N<|\Omega|$, since we always assume that $p',p$ are absolutely continuous with respect to $q$, which allows to redefine $\Omega$ to only contain the $N$ options covered by $q$. 

For the proof of the second property in \eqref{uniformmonotonicity}, we let the random variable $X$ indexing the partitions $E_1,E_2$ of $\Omega$, where $E_1$ denotes the support of $u_{N'}$ and $E_2=\Omega{\setminus} E_1$ its complement, and $Y$ representing the choice inside of the selected partition $E_i$ given $X=i$. Letting $q(Z) = u_N$, and $p(Z) = u_{N'}$, then it follows from 
addivity under coarse-graining that $C_{q(X)}(p(X)) = C_{u_N}(u_{N'}) = f(N',N)$, and letting $p(Z) = u_{M}$, we obtain 
\[
f(M,N) = C_{q(X)}(p(X)) + C(Y|X=1) = f(N',N) + f(M,N') \, ,
\]

\noindent since $p(X{=}1) = 1$, $p(X{=}2) = 0$, and $C(Y|X=1) = C_{u_{N'}}(u_M)$. Therefore it follows that $f(M,N)\geq f(M,N')$.
\end{proof}

\medskip

\begin{Proposition}[Prop.~\ref{brdmproc}] \label{prop:boundedoptDM}
If $(p_\beta)_{\beta\geq 0}$ is a family of bounded-optimal posteriors given by \eqref{eq:unconstrained} with cost function $C_q(p) = D_{KL}(p\|q)$, then $\beta$ is a resource parameter with respect to any cost function, in particular
\begin{equation}\label{BRprec}
q = p_0 \, \precnsim_q \, p_{\beta'} \, \precnsim_q \, p_{\beta} \qquad \forall \beta',\beta \ \ \text{with} \ \ \beta'<\beta \,.
\end{equation}
\end{Proposition}

\begin{proof}
Part of the proof generalizes a result in \cite{Rossignoli2004} to the case of a non-uniform prior $q$, by making use of our new characterization $(v)$ of $\prec_q$, by which it suffices to show that $\beta \mapsto \sum_{i=1}^{l-1} (p_{\beta,i})^{\downarrow} + a_q(k,l) (p_{\beta,l})^{\downarrow}$ is an increasing function of $\beta$ for all $k,l$ specified in Theorem \ref{thm:charrel}. By \eqref{boltzmann} below, we have 
\[
\partial_\beta p_{\beta,i} = \partial_\beta \big(\tfrac{1}{Z_\beta} \, q_i \, e^{\beta U_i} \big) = p_{\beta,i}(U_i - \mathbb E_{p_\beta}[U]) \, ,
\] 

\noindent where $U_i$ are the decreasingly arranged utility values $U(x)$ for $x\in\Omega$, so that also $p_{\beta,i}/q_i$ is arranged decreasingly. From the ordering of the $U_i$, it is easy to see that
\[
\left(\sum_{i=1}^k p_{i} U_i + p_{k+1} U_{k+1}\right) \sum_{j=1}^k p_j \ \leq \ \sum_{i=1}^k p_i U_i \, \left( \sum_{j=1}^k p_j + p_{k+1} \right)
\]

\noindent with the notation $p_k\coloneqq p_{\beta,k}$, from which it follows that $S_k\coloneqq \sum_{i=1}^k \hat p_k U_i$ with $\hat p_k \coloneqq {p_k}/{\sum_{j=1}^kp_j}$, is monotonically decreasing in $k$ (with $S_N=\mathbb E_{p}[U]$), and therefore 
\[
\sum_{i=1}^k p_{i}(U_i - \mathbb E_{p}[U]) = \left(\sum_{i=1}^k \hat p_i U_i -\mathbb E_{p}[U] \right) \sum_{j=1}^k p_j = (S_k - S_N) \sum_{j=1}^k p_j  \ \geq \ 0
\] 

\noindent for all $k\leq N$. Hence, it suffices to show that \smash{$\sum_{i=1}^{l-1} x_i + t x_l \geq 0$} if \smash{$\sum_{i=1}^k x_i \geq 0$} for all $k$ and $t\in[0,1]$. If $x_l\geq 0$, there is nothing to show, and if $x_l<0$, we have \smash{$\sum_{i=1}^{l-1}x_i + t x_l \geq \sum_{i=1}^{l}x_i\geq 0$}, which completes the proof of $p_{\beta'}\prec_q p_\beta$. 

It remains to see that $p_\beta \not\prec_q p_{\beta'}$. This follows again from $(v)$ in Theorem \ref{thm:charrel}, more precisely from the requirement that $p_{\beta,1}\leq p_{\beta',1}$ if $p_\beta\prec_q p_{\beta'}$, after establishing that \smash{$\beta\mapsto Z_\beta^{-1}e^{\beta U_1}$} is monotonically increasing. For the latter, note that since the $U_i$ are ordered decreasingly, \smash{$\sum_{i=1}^N U_i q_i e^{\beta U_i} \leq U_1 \sum_{i=1}^N q_i e^{\beta U_i}$} from which follows that $\partial_\beta(Z_\beta^{-1}e^{\beta U_1}) \geq 0$, which completes the proof.
\end{proof}

\end{appendix}

\bigskip

\section*{Acknowledgement}

This study was funded by the European Research Council (ERC-StG-2015-ERC Starting Grant, Project ID: 678082, ``BRISC: Bounded Rationality in Sensorimotor Coordination'').


\end{document}